\documentclass[journal,comsoc]{IEEEtran}
\usepackage{amsfonts}
\usepackage{amssymb}
\usepackage{cite}
\usepackage[cmex10]{amsmath}
\usepackage{float}
\usepackage{color}
\usepackage{stfloats,fancyhdr}
\usepackage{amsmath}
\usepackage{algorithmic}
\usepackage{multirow}
\usepackage{changepage}
\usepackage[normalem]{ulem}
\usepackage{amsthm}
\usepackage[ruled,linesnumbered]{algorithm2e}
\usepackage{makecell}

\newtheorem{proposition}{Proposition}

\newtheorem{remark}{Remark}

\IEEEoverridecommandlockouts

\ifCLASSINFOpdf
  \usepackage[pdftex]{graphicx}
  \DeclareGraphicsExtensions{.pdf,.jpeg,.png}
\else
  \usepackage[dvips]{graphicx}
  \DeclareGraphicsExtensions{.pdf}
\fi

\usepackage{subfigure}

\usepackage{fancybox,dashbox}

\begin{document}
%
\title{Graph Neural Networks for Distributed Power Allocation in Wireless Networks: Aggregation Over-the-Air
\thanks{This work was supported in part by the National Natural Science Foundation of China under Grant 62201361, in part by the Guangdong Provincial Department of Science and Technology under Project 2020B1212030002, in part by the National Natural Science Foundation of China under Grant 62101293, and in part by the Shenzhen Science and Technology Innovation Commission, China. (\textit{Corresponding author: Zhi Quan.})}
\thanks{Y. Gu and Z. Quan are with School of Electrical and Information Engineering, Shenzhen University, Shenzhen 518060, Guangdong, China (email: yifan.gu@szu.edu.cn, zquan@szu.edu.cn).}
\thanks{C. She is with School of Electrical and Information Engineering, The University of Sydney, Sydney, NSW 2006, Australia (email: shechangyang@gmail.com).}
\thanks{C. Qiu is with Department of Mathematics and Theories, Peng Cheng Laboratory, Shenzhen 518066, Guangdong, China (email:qiuch@pcl.ac.cn).}
\thanks{Xiaodong Xu is with the State Key Laboratory of Networking and Switching Technology, Beijing University of Posts and Telecommunications, Beijing 100876, China, and also with the Department of Broadband Communication, Peng Cheng Laboratory, Shenzhen 518066, Guangdong, China (e-mail: xuxiaodong@bupt.edu.cn).}
}
\author{Yifan Gu, Changyang She, Zhi Quan, Chen Qiu, and Xiaodong Xu
}

\maketitle

\begin{abstract}
Distributed power allocation is important for interference-limited wireless networks with dense transceiver pairs. In this paper, we aim to design low signaling overhead distributed power allocation schemes by using graph neural networks (GNNs), which are scalable to the number of wireless links. We first apply the message passing neural network (MPNN), a unified framework of GNN, to solve the problem. We show that the signaling overhead grows quadratically as the network size increases. Inspired from the over-the-air computation (AirComp), we then propose an Air-MPNN framework, where the messages from neighboring nodes are represented by the transmit power of pilots and can be aggregated efficiently by evaluating the total interference power. The signaling overhead of Air-MPNN grows linearly as the network size increases, and we prove that Air-MPNN is permutation invariant. To further reduce the signaling overhead, we propose the Air message passing recurrent neural network (Air-MPRNN), where each node utilizes the graph embedding and local state in the previous frame to update the graph embedding in the current frame. Since existing communication systems send a pilot during each frame, Air-MPRNN can be integrated into the existing standards by adjusting pilot power. Simulation results validate the scalability of the proposed frameworks, and show that they outperform the existing power allocation algorithms in terms of sum-rate for various system parameters.
\end{abstract}



\IEEEpeerreviewmaketitle
\section{Introduction}
With the ever-growing demands of connections in the next generation of wireless networks, the conflicts between the massive connections and limited spectrum resource is getting increasingly severe \cite{spectrum}. Transmit power allocation is one of the key enabler to support massive connections under the spectrum scarcity \cite{powerallocation}. Consider a wireless network with dense transceiver pairs and full spectrum reuse, e.g., Device-to-Device (D2D) \cite{D2D2}, Ad-Hoc \cite{adhoc}, and ultra-dense cellular \cite{ultradense} networks, where the neighbour links can generate excessive interference to each other if they use the maximum transmit power. By optimizing the power allocation, it is possible to adjust the transmit power of each link according to the dynamics of wireless channels, such that a certain utility function is maximized. Nevertheless, when multiple transceiver pairs share a single frequency band, the power allocation problem for sum-rate maximization can be very challenging, which is NP-hard and non-convex in general \cite{NP-hard1}.

A well-known power allocation algorithm, i.e., weighted minimum mean squared error (WMMSE) \cite{WMMSE,WMMSE1}, is a centralized iterative algorithm. To implement WMMSE in large-scale wireless networks, the computing complexity is high and a large number of iterations is needed to converge to a stationary point. As a centralized algorithm, WMMSE needs the channel state information (CSI) of all the interference links. It leads to high signaling overhead in practical systems.

Inspired by the recent advances in data-driven models, deep learning-based methods were proposed to solve the power allocation problems in wireless networks to combat the computation complexity \cite{DL1,DL3,DL4,DL5,DL6}. Specifically, deep neural networks (DNNs) were adopted to learn the mapping from the channel dynamics to the optimal policy, such that the online computation time of the policy can be significantly reduced. For example, the authors in \cite{DL1} approximated the power allocation policy obtained from the WMMSE algorithm by a fully-connected deep neural network (FNN). When using FNNs in wireless networks, there are two issues: 1) FNN does not exploit the network topology. When the topology of a wireless network becomes different, we need to re-train the FNN. 2) FNN is not scalable to the number of nodes in a wireless network. When the number of wireless links is dynamic, we need to change the structure of the FNN, such as the dimensions of the input layer and the hidden layers \cite{DL4}. In large-scale dynamic wireless networks, the nodes can be dense and the environment is highly dynamic, which restricts the implementation of FNN \cite{Cshe1}.

To obtain scalable solutions for dynamic wireless networks, graph neural networks (GNNs) \cite{GNN1} were introduced to solve the power allocation problems of large-scale dynamic wireless networks, where each communication link is represented by a node and each interference link is represented by an edge. The GNN can extract the neighbour information of a node in a graph according to the network topology and update the representation of the node, i.e., graph embedding. The graph embedding is a low-dimensional vector representing the whole graph or part of the graph, and we refer the graph embedding as the representation of each node in this paper \cite{Distancec_based}. After several rounds of updates, the graph embedding contains information of the neighbour network multiple-hops away from the node, and can be used to determine the optimal resource allocation policies \cite{Cshe2,Cshe3}. Among different GNN architectures, spatial convolutional graph neural network \cite{GCN} is one of the most widely used architectures in solving the power allocation problems for wireless networks \cite{GCNwirelessreview1}. In the following, we will consider the spatial GNNs and refer them as GNN for short.

Several centralized GNN-based power allocation policies were developed in the existing literature. In \cite{REGNN}, the authors considered a large-scale dynamic wireless network with a set of transmitters and receivers, and formulated a generic resource allocation problem. Based on the GNN, a random edge graph neural network (REGNN) framework was introduced to solve the problem where the fading channels were modelled as the adjacency matrix of a graph. A prime-dual learning algorithm was designed to train the weights of REGNN in an unsupervised manner. The authors in \cite{Power_allocation_multi_cell} considered the power allocation problem for a multi-cell-multi-user wireless network. With the fact that the features of base stations and users are different, a heterogenous GNN (HGNN) was adopted to model the system and optimize the sum-rate performance. 

Compared with the above centralized policies, distributed polices are more practical in dynamic large-scale wireless networks \cite{distributed_benefit}. The authors in \cite{unfold} proposed a distributed unfold WMMSE (UWMMSE) framework to depart conventional WMMSE algorithm based on the GNN. Compared with the WMMSE, the scalability of the proposed UWMMSE was greatly improved. Distributed power allocation policies for HGNN was studied in \cite{HGCN1}, where a wireless network consisting of multiple transmitters and receivers with different numbers of antennas was considered. The authors in \cite{V2V} adopted the deep reinforcement learning (DRL) framework to solve the resource allocation problem of a vehicular wireless network in a distributed manner. They defined the observation for the DRL as the graph embedding in the GNN that aggregates the network state information at each node. The learned policy can scale well to different network sizes. In \cite{YShen_CGNET_conference,Yshen_CGNET_trans}, the authors proposed the wireless channel graph convolution network (WCGCN) to solve the power allocation problems of wireless networks, and developed policies that can be implemented in a distributed manner. Very recently, the authors in \cite{distributed_GNN} considered the wireless channel impairments in the distributed implementation of GNN, and analyzed the performance of a GNN-based binary classifier for both uncoded and coded wireless communication systems. To implement the GNN-based policies in a distributed manner, the aforementioned works utilized the framework of message passing neural network (MPNN). MPNN is a unified framework of GNN that enables the graph convolution operation by updating the graph embedding via message passing and aggregation \cite{MPNN}.
\subsection{Motivation}
Based on the aforementioned GNN-based frameworks, it is possible to obtain distributed and scalable power allocation policy in large-scale dynamic wireless networks. Nevertheless, as shown in \cite{unfold,HGCN1, YShen_CGNET_conference,Yshen_CGNET_trans,V2V,distributed_GNN}, the signaling overhead of executing those policies can grow fast as the network size increases. The signaling overhead may include the CSI estimation overhead and the message passing overhead between the nodes. To execute the graph convolution operation in a distributed manner, the nodes need to firstly obtain the CSI of multiple wireless links, including the direct communication link and the neighbour interference links \cite{YShen_CGNET_conference,Yshen_CGNET_trans}. After that, the nodes need to exchange their graph embeddings multiple times. The signaling overhead can be dominant and may degrade the system performance significantly for large-scale wireless networks. In this context, our work is driven by the following question: \textit{Can we design a distributed low-complexity power allocation policy that is scalable to the number of wireless links and has low signaling overhead?} To answer this critical question, in this paper, we formulate a generic power allocation problem, and propose the novel Air-MPNN framework that exploits the interference signals in message passing.

The idea of the proposed Air-MPNN is inspired by the concept of the over-the-air computation (AirComp) \cite{AirComp} and the MPNN framework. In AirComp systems, the devices can transmit messages simultaneously by using non-orthogonal wireless resources, and the server can decode an arithmetic of the messages directly from the superimposed signal, such as sum, mean or max. To be more specific, the nodes need to acquire the CSI of the wireless links and generate messages to compensate the CSI in a symbol-wised manner in order to calculate the arithmetic of the aggregated messages. Different from AirComp, the proposed Air-MPNN framework only needs the aggregation of the neighbor nodes embeddings and edge features, i.e., CSI of interference links. To achieve this, Air-MPNN does not need to decode the symbols from the messages as in AirComp, and only needs to estimate the total interference power from all the interference links. Considering the fact that the pilots from the interference links carry the CSI of the interference links, we can thus treat the aggregation of pilots as the aggregation of the features of the interference links. Since the aggregated messages are represented by the aggregated interference, the proposed Air-MPNN can accomplish the graph convolution operation without knowing the individual CSI of the interference link. To update the graph embedding, only a limited number of message exchanges are needed.

To further reduce the signaling overhead of the Air-MPNN framework, we exploit the temporal correlation of network states in wireless networks. For example, the subframe duration in 5G New Radio (NR) is 1~ms and the maximum frame duration in Wi-Fi 802.11ac is 5.5~ms, where the channel coherence time can be longer than $10$~ms. Therefore, it is possible to utilize the graph embeddings in the previous frame to update the embeddings in the current frame. In this way, the nodes only need to perform the graph convolution operation once in the current frame. Inspired by the conventional recurrent neural network (RNN), we develop a message passing recurrent neural network (MPRNN). With Air-MPRNN, we can further reduce the signaling overhead, and it can be potentially integrated into the existing wireless networks without changing the frame structures in the standards since all the nodes only need to send the pilot signal once. To summarize the evolution from MPNN to Air-MPNN, and then to Air-MPRNN, the MPNN requires the nodes to broadcast pilot one by one for multiple times, while the Air-MPNN just requires all the nodes to broadcast pilot simultaneously by multiple times. Finally, the Air-MPRNN only needs all the nodes to broadcast the pilot simultaneously once. To the best of our knowledge, this is the first paper to propose novel GNN-based architectures by exploiting the interference signals in message passing and temporal correlation of network states to reduce signaling overhead.
\subsection{Contributions}
The main contributions of this paper are summarized as follows:
\begin{itemize}
  \item We formulate a generic power allocation problem in a large-scale wireless network with dense communication links and strong interference. Without loss of generality, we take the weighted sum-rate maximization problem as an example to find the optimal policy. We show that when implementing MPNN-based policy in a distributed manner, the signaling overhead for CSI estimation grows quadratically as the network size increases, and the overhead for message passing grows linearly with the size of the wireless network and the number of layers of the GNN.
  \item We design a novel Air-MPNN framework, where the aggregated messages are represented by the total interference power. We prove that the total interference power is a permutation invariant operation of node embeddings and edge features by showing that it is a special case of sum. We also show that the signaling overhead for CSI estimation of the proposed Air-MPNN framework grows linearly as the network size increases. In terms of the overhead for message passing, all the nodes only need to broadcast the generated message once for updating the graph embedding of each layer of the GNN.
  \item We then propose the Air-MPRNN framework by combining the Air-MPNN with RNN to capture the temporal information of wireless networks. The proposed Air-MPRNN framework can estimate the CSI of the communication links and aggregate the messages from interference links in one phase. During each frame, all the nodes can obtain the graph embedding of the output layer and the local CSI efficiently by only broadcasting the pilot signal once, where the transmit power of each pilot is determined by the graph embedding and local CSI in the previous frame.
  \item Through extensive simulation results\footnote{The codes to reproduce the simulation results are available on https://github.com/Yifan-Gu-SZU/GNN-aggregation-over-the-air.}, we illustrate that the proposed Air-MPNN and Air-MPRNN frameworks can outperform the existing GNN-based framework in terms of signaling overhead. They also outperform WMMSE and equal power allocation (EPA) in terms of the sum-rate when the signaling overhead is considered. In addition, we evaluate the impact of various system parameters on the performance, and validate the scalability of the proposed frameworks under different network sizes and link densities.
\end{itemize}
\textit{Notation}: We use lower-case letters to denote scalar numbers, while bold lower case, and upper case letters to represent vectors and matrixes, respectively. We let ${a^*}$ to be the conjugate of the complex scalar $a$. ${\bf A}^T$ and ${\bf A}^H$ denote the transpose, and Hermitian transpose of matrix ${\bf A}$, respectively. $\mathbb{C}$ is the set of complex numbers. $\left|\cdot\right|$ and $\left\|\cdot\right\|$ refer to the absolute value of a complex scalar, and the Euclidean vector norm, respectively. ${\cal{CN}}\left(\mu,\sigma^2\right)$ represents the circularly symmetric complex Gaussian distribution with mean $\mu$ and variance $\sigma^2$. We use $\mathbb{E}$ to denote the expectation operation and $\cal{O}$ to be the big $O$ notation. We let ${\left[ {\cdot} \right]_{ij}}$ to be the element of the $i$-th row and $j$-th column in a two-dimensional matrix.
\section{System Model and Problem Formulation}
\subsection{System Model}
\begin{figure}[!h]
	\centerline{\includegraphics[width=0.45\textwidth]{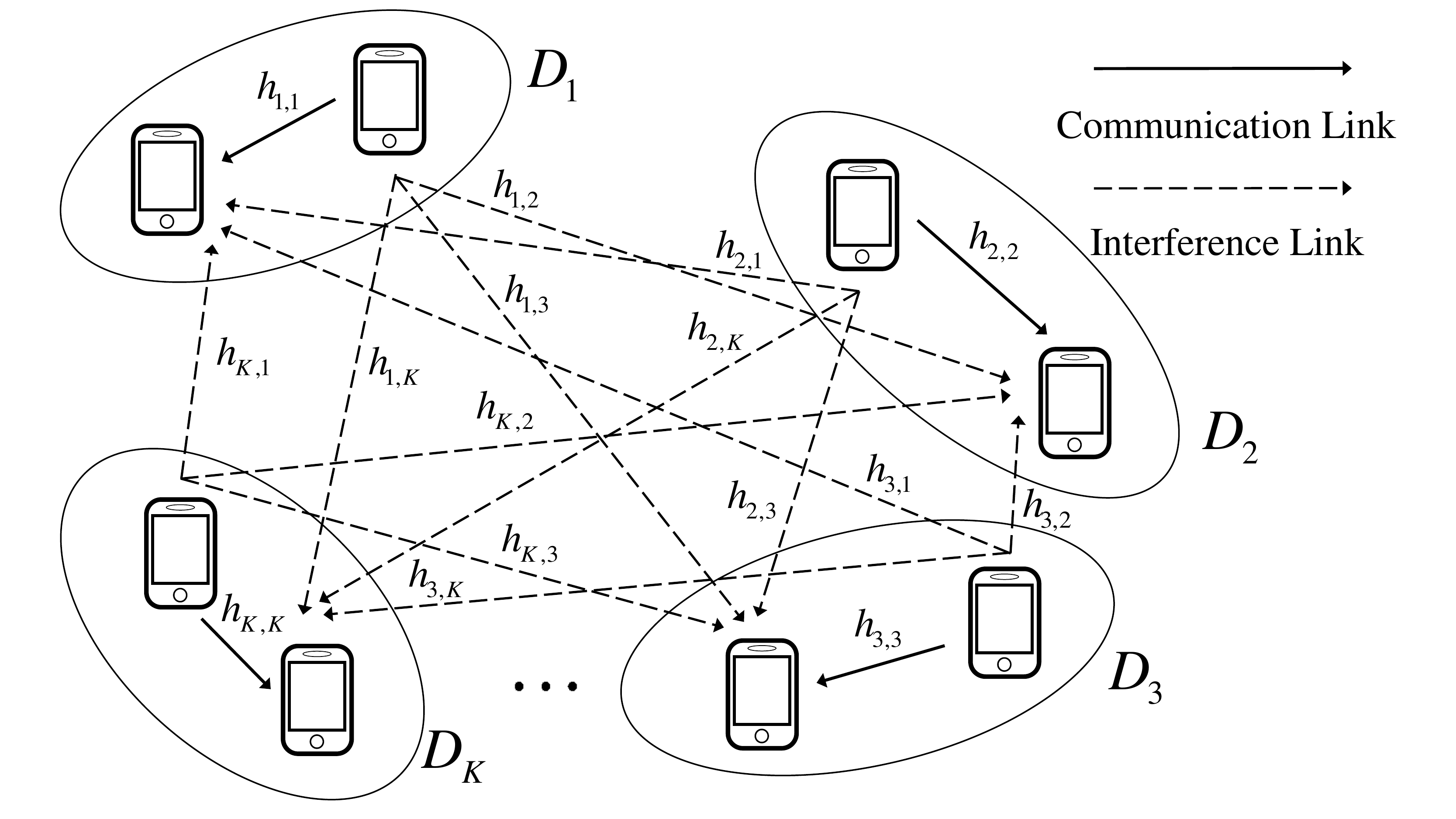}}
	\caption{The considered wireless network with $K$ D2D pairs.}
	\label{fig_system_model}
\end{figure}
As depicted in Fig. \ref{fig_system_model}, we consider a D2D network consisting of $K$ randomly located transceiver pairs $D_1,D_2,\cdots,D_K$, and it is not hard to extend our method to other types of networks, such as Ad-Hoc networks, and ultra-dense cellular networks. We assume that all the D2D pairs share the same spectrum and both the transmitter and the receiver of each D2D pair are single-antenna devices. We use $h_{i,i}$ to denote the channel coefficient of the direct D2D communication link for $D_i$, and $h_{i,j}$ to represent the channel coefficient of the interference link between $D_i$ and $D_j$, $i,j \in \left\{1,2,\cdots,K\right\}$. We assume that the channel coefficient remains unchanged during each frame and changes across different frames. We use ${{{\bf{H}}_I}} \in {\mathbb{C}^{K \times K}}$ to represent the channel coefficient matrix of all the interference links, with ${\left[ {{{\bf{H}}_I}} \right]_{ij}} = {h_{ij}}$ if $i \ne j$ and ${\left[ {{{\bf{H}}_I}} \right]_{ij}} = 0$ if $i=j$. Moreover, we use ${{\bf{z}}_i} \in {\mathbb{C}^{1 \times L}}$ to denote the local state information of the $i$-th D2D communication link, including the channel coefficient of the D2D link $h_{i,i}$, the weight $w_i$, and etc.. Note that we can assign larger weights to D2D pairs with higher priority or longer link distance. In this way, its priority or fairness is guaranteed. For notation simplicity, we use ${\bf{Z}} \in {\mathbb{C}^{K \times L}} = {\left[ {{{\left( {{{\bf{z}}_1}} \right)}^T},{{\left( {{{\bf{z}}_2}} \right)}^T}, \cdots ,{{\left( {{{\bf{z}}_K}} \right)}^T}} \right]^T}$ to represent the local state information for all the D2D links. With the above definitions, the global network state information of the wireless network is given by $\left\{{{{\bf{H}}_I}},{\bf{Z}} \right\}$.
\begin{figure}[!h]
	\centerline{\includegraphics[width=0.45\textwidth]{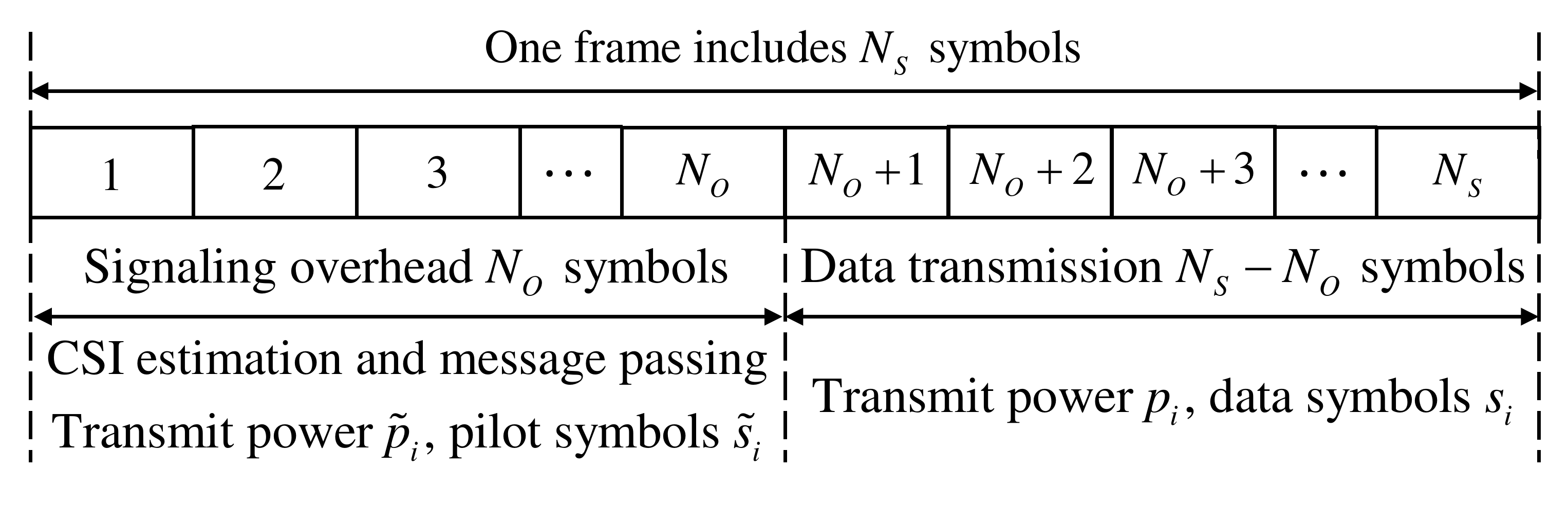}}
	\caption{The frame structure of the considered D2D network.}
	\label{fig_frame_structure}
\end{figure}

Time is discretized into frames. As shown in Fig. \ref{fig_frame_structure}, each wireless link can transmit a total number of $N_S$ symbols in one frame, where $N_O$ symbols are used for transmitting a pilot, and $N_S-N_O$ symbols are used for data transmission. Specifically, the pilot represents the demodulation reference signal for channel estimation in the 5G NR systems, or the long training field (LTF) in the Wi-Fi systems. To execute GNN-based policy in the wireless network, the signaling overhead should also include the symbols for message passing. At the beginning of each frame, we let $o_i$ to denote the local observation at the $i$-th D2D pair in order to execute the policy in a distributed manner. Since the value of $N_O$ and the definition of $o_i$ depend on specific policies, we will discuss them later.
We denote the pilot sequence used by the $i$-th D2D pair by $\tilde{s}_i$. The transmit power for sending $\tilde{s}_i$ is $\tilde{p}_i$. After CSI estimation and message passing, each D2D pair can determine a local transmit power $p_i$ for data transmission during the remaining data transmission phase. The data symbols of the $i$-th D2D pair is denoted by~$s_i$.
\subsection{Generic Power Allocation Problem Formulation}
A generic power allocation policy is a mapping from the global network state information to the transmit power, i.e., normalized transmit power, of all the D2D links for data transmission, i.e., ${\bf{p}} = f\left( {{{\bf{H}}_I},{\bf{Z}}} \right)$. We use $r_i\left( {{\bf{p}};{{\bf{H}}_I},{\bf{Z}}} \right)$ to represent the reward function of the $i$-th D2D pair. The objective function is the expected weighted sum of rewards $\bar{r} = \mathbb{E}\left[ \sum\limits_{i = 1}^K {w_ir_i\left( {{\bf{p}};{{\bf{H}}_I},{\bf{Z}}} \right)}  \right]$, where the expectation is over the global network state information $\left\{{{{\bf{H}}_I}},{\bf{Z}}\right\}$. Examples of rewards can be rate, decoding error probability, energy efficiency and etc.. Formally, the power allocation problem can be described~as
\begin{equation}\label{resource_management_problem}
\begin{split}
  &{f}\left( {{{\bf{H}}_I},{\bf{Z}}} \right) =  \arg \max{\bar{r}},  \\
  & \text{s.t.} \quad \bar{r} = \mathbb{E}\left[ \sum\limits_{i = 1}^K {w_ir_i\left( {{\bf{p}};{{\bf{H}}_I},{\bf{Z}}} \right)}  \right],  \\
  & \quad \quad 0 \le p_i \le1, \forall i, \\
\end{split}
\end{equation}

Different from the existing works that did not consider the signaling overhead in executing the policy, we are aiming to find a policy with very limited signaling overhead, such that the proposed solutions can be implemented in large-scale dynamic wireless networks. To be more specific, we use different types of GNNs to represent the optimal policy ${f}\left( {{{\bf{H}}_I},{\bf{Z}}} \right)$ and train these GNNs in a centralized manner. By carefully designing the structures of these GNNs, the inference can be executed in a distributed manner with local observation and very low signaling overhead. In other words, after centralized training, the GNN-based policies can be executed by each D2D pair without a centralized controller.
\subsection{Weighted Sum-Rate Maximization Problem Example}
In this subsection, we take a weighted sum-rate maximization problem as an example of problem (\ref{resource_management_problem}). In this example, the local state of each D2D communication link $D_i$ is defined as
\begin{equation}\label{zk1}
{{\bf{z}}_i} = \left[ {{h_{i,i}},{w_i}}\right],
\end{equation}
where $h_{i,i}$ and $w_i$ are the channel coefficient, and the weight of the $i$-th D2D link, respectively. The received signal of the $i$-th D2D pair during data transmission phase can be expressed~as
\begin{equation}\label{signal}
y_i = \sqrt {p_i} {h_{i,i}}{s_i} + \sum\limits_{j = 1,j \ne i}^K {\sqrt {p_j}} {h_{j,i}}{s_j}  + {n_i},
\end{equation}
where $n_i \sim {\cal{CN}}\left(0,\sigma^2\right)$ is the additive white Gaussian noise (AWGN) at the receiver side. From the received signal, the signal-to-interference-plus-noise ratio (SINR) for $D_i$ can be expressed~as
\begin{equation}\label{SINR}
{\xi _i}\left({{\bf{p}},{\bf{H}}_I,{\bf{Z}}}\right) = {{{p_i}{{\left| {{h_{i,i}}} \right|}^2}} \over {\sum\limits_{j = 1,j \ne i}^K {{p_j}{{\left| {{h_{j,i}}} \right|}^2} + {\sigma ^2}} }}.
\end{equation}
The reward of the $i$-th D2D pair is given by the Shannon's capacity, i.e.,
\begin{equation}\label{rewardrate}
\begin{split}
r_i\left( {{\bf{p}};{{\bf{H}}_I},{\bf{Z}}} \right) & = {{{N_S - N_O} \over {N_S}}{{\log }_2}\left( {1 + {\xi _i}\left({{\bf{p}},{\bf{H}}_I,{\bf{Z}}}\right)} \right)}  \\
&= {{{N_S - N_O} \over {N_S}}{{\log }_2}\left( {1 + {{{p_i}{{\left| {{h_{i,i}}} \right|}^2}} \over {\sum\limits_{j = 1,j \ne i}^K {{p_j}{{\left| {{h_{j,i}}} \right|}^2} + {\sigma ^2}} }}} \right)}.
\end{split}
\end{equation}
By defining the total reward of the wireless network as the weighted sum-rate, the problem can be captured by (\ref{resource_management_problem}).

\section{MPNN for Distributed Power Allocation}
In this section, we first adopt the existing MPNN framework to solve problem (\ref{resource_management_problem}). Based on this framework, we discuss the observation required for each D2D pair to execute the policy in a distributed manner, and analyze the signaling overhead for obtaining the observation, i.e., $N_O^{\rm mpnn}$.
\begin{figure*}[!t]
	\centerline{\includegraphics[width=0.8\textwidth]{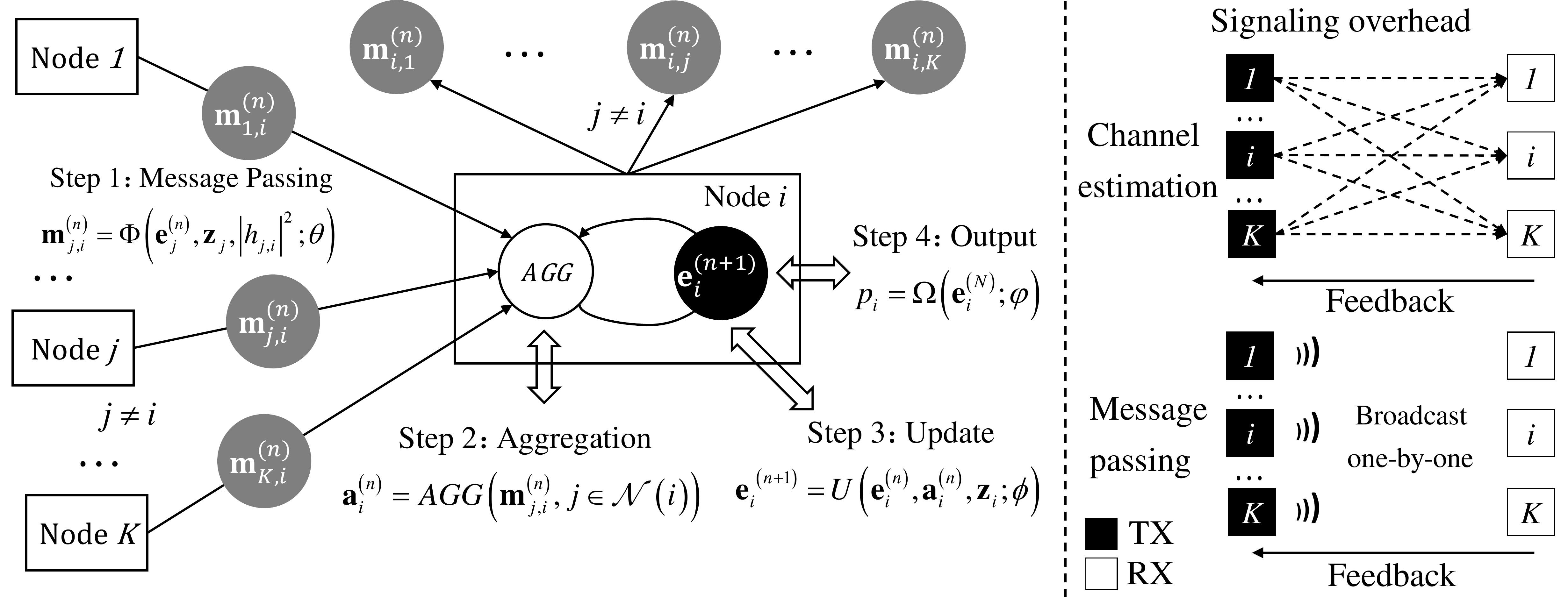}}
	\caption{The illustration of the existing MPNN framework.}
	\label{fig_MPNN}
\vspace*{4pt}
\end{figure*}
\subsection{Graph Model}
We represent the wireless network by a fully connected graph $G = \left( {V,{\cal E}} \right)$. It is composed of a set of nodes $V$, and a set of edges $\cal E$. A node $v_i \in V$ represents a D2D pair $D_i$, and an edge $\varepsilon _{i,j} \in {\cal E}$ defines the interference link from $D_i$ to $D_j$. The features of the node $v_i$ involves the local state information of the $i$-th D2D communication link, e.g., ${{\bf{z}}_i} = \left[ {\left|{h_{i,i}}\right|^2,{w_i}} \right]$. The feature of the edge $\varepsilon_{i,j}$ is defined as the channel power gain $\left| {{h_{i,j}}} \right|^2$ of the interference link from the $i$-th D2D pair to the $j$-th D2D pair. Note that the features of all the nodes and edges are included in the global network state information $\left\{{{\bf{H}}_I},{\bf{Z}}\right\}$. We choose $\left| {{h_{i,j}}} \right|^2$ instead of ${{h_{i,j}}}$ as the feature because the channel coefficient has both real, and imaginary parts that may increase the complexity of the model. For the power allocation problem formulated in Sec. II-C, the channel power gain is sufficient for determining the optimal power for data transmission~\cite{goldsmith2005wireless}.
\subsection{Power Allocation via MPNN}
The convolution kernel of MPNN is designed to have a message passing phase, an aggregation phase and an update phase. For a better understanding, we depict an illustration of the graph convolution operation in the MPNN framework in Fig. \ref{fig_MPNN} on top of this page. During the message passing phase, each node generates messages based on its local feature and edge features of the neighbours, and passes the generated messages to different neighbours. In the aggregation phase, the node aggregates the messages from all of its neighbours to generate an aggregated information. Next in the update phase, the node considers both the aggregated information from its neighbours and its local features to update its embedding. The embedding obtained by one step of update is a representation of its local features, and the features of its neighbors one-hop away. Formally, the message passing, aggregation, and update phases of the MPNN framework can be expressed as
\begin{equation}\label{message}
\text{Message passing:}\quad {\bf m}_{j,i}^{\left( n \right)} = \Phi \left( {{\bf {e}} _j^{\left( n \right)},{\mathbf{z}_j},\left| {{h_{j,i}}} \right|^2;\theta } \right),
\end{equation}
\begin{equation}\label{agg}
\text{Aggregation:}\quad  {\bf {a}} _i^{\left( n \right)} = AGG\left( {{\bf m}_{j,i}^{\left( n \right)},j \in {\cal N}\left( i \right)} \right),
\end{equation}
\begin{equation}\label{update}
\text{Update:}\quad{{\bf {e}}_i}^{\left( {n + 1} \right)} = U\left( {{\bf {e}}_i^{\left( n \right)},{\bf {a}}_i^{\left( n \right)},{\mathbf{z}_i};\phi } \right),
\end{equation}
where ${\bf m}_{j,i}^{\left( n \right)}$ is the message from node $v_j$ to node $v_i$ in the $n$-th layer of graph convolution, and ${\cal N}\left( i \right)$ denotes the set of neighbour nodes of $v_i$. ${\bf {a}}_i^{\left( n \right)}$ is the aggregated message from all the neighbours of $v_i$, and ${\bf {e}}_i^{\left( n \right)}$ is the embedding of node $v_i$. Both ${\bf {a}}_i^{\left( n \right)}$ and ${\bf {e}}_i^{\left( n \right)}$ are the features in the $n$-th layer of the GNN. For the $0$-th layer, the graph embedding can be initialized to zero, i.e., ${\bf {e}}_i^{\left( 0 \right)}={\bf0}$. To ensure permutation invariance, i.e., local invariance, we can use summation, mean value, or maximum value of the inputs for the aggregation function, $AGG \in \left\{sum, mean, max\right\}$. With these functions, the output of $AGG\left(.\right)$ does not change with the order of the inputs. The message function $\Phi \left( {\cdot;\theta } \right)$, and update function $U\left( {\cdot;\phi } \right)$ are multi-layer perceptrons (MLPs) with parameters $\theta$ and $\phi$, respectively. After $N$ rounds of updates, the final embedding of each node ${\bf {e}}_i^{\left( N \right)}$ is utilized to determine its optimal transmit power for data transmission. Let $\Theta = \left\{\theta,\phi\right\}$ represents the parameters of the MPNN model, the power allocation problem (\ref{resource_management_problem}) can be reformulated as
\begin{equation}\label{problem2}
\begin{split}
  &{\Theta, \varphi} = \arg \max {\bar{r}},  \\
  & \text{s.t.} \quad {\bf {e}}_i^{\left( N \right)} = {\rm MPNN}\left( {{{\bf{H}}_I},{\bf{Z}};\Theta } \right),\forall i,  \\
  & \quad \quad {p_i} = {\Omega \left( {{\bf {e}}_i^{\left( {N } \right)};\varphi } \right)},\forall i,  \\
  & \quad \quad  \bar{r} = \mathbb{E}\left[ \sum\limits_{i = 1}^K {w_ir_i\left( {{\bf{p}};{{\bf{H}}_I},{\bf{Z}}} \right)}  \right], \\
\end{split}
\end{equation}
where ${\rm MPNN}\left(.\right)$ is the model with the graph convolution kernel defined in (\ref{message})-(\ref{update}). To ensure the maximum transmit power constraint given in (\ref{resource_management_problem}), we can adopt an MLP $\Omega\left(\cdot;\varphi \right)$ with a sigmoid activation function at the output layer to infer the transmit power from the embedding, such that the transmit power is normalized between 0 and 1. In problem (\ref{problem2}), we use the MPNN model with parameters $\Theta$ and $\varphi$ to approximate the optimal power allocation function ${f}\left( . \right)$ defined in problem (\ref{resource_management_problem}). The parameters can be optimized by stochastic gradient descend, where we use a batch of samples for different realizations of the network state information by $\left\{{{{\bf{H}}_I}},{\bf{Z}} \right\}$ to estimate the gradient, i.e., ${\nabla _{\Theta ,\varphi }}\mathbb{E}\left[ { - \sum\limits_{i = 1}^K {{w_i}{r_i}\left( {{\bf{p}};{{\bf{H}}_I},{\bf{Z}}} \right)} } \right]$.
\subsection{Signaling Overhead of Distributed Implementation via MPNN}
The above power allocation policy is trained in a centralized manner. In this section, we define the observation required for each node for distributed inference. To compute (\ref{message})-(\ref{update}), each node needs to obtain the observation given by
\begin{equation}\label{observationMPNN}
o_i^{{\rm{mpnn}}} = \left\{ {{h_{i,i}},{w_i},\underbrace {\left\{ {{h_{j,i}},\forall j \ne i} \right\},\left\{ {{\bf{e}}_j^{\left( n \right)}, {\mathbf{z}_j}, \forall j \ne i,n} \right\}}_{{\rm Determine}\;{\bf{m}}_{j,i}^{\left( n \right)}\;{\rm in \;(6)\;and\;compute\;}{\bf{a}}_i^{\left( n \right)}\;{\rm in\;(7)}}} \right\}.
\end{equation}
Note that in (\ref{observationMPNN}), each node $v_i$ needs the information ${\left\{ {{h_{j,i}},\forall j \ne i} \right\}}$ and $\left\{ {{\bf{e}}_j^{\left( n \right)}, {\mathbf{z}_j}, \forall j \ne i,n} \right\}$ to determine the messages from its neighbours and compute the aggregated message ${{\bf{a}}_i^{\left( n \right)}}$ for each layer of GNN. Then, the node can update its local embedding with the aggregated message and its local state.
\subsubsection{Signaling Overhead for CSI Estimation}According to (\ref{observationMPNN}) and Fig. \ref{fig_MPNN}, each node $v_i$ needs the CSI of $K$ links, including its D2D communication link $h_{i,i}$, and the interference links $h_{j,i}$, $1 \le j \le K$, $j \ne i$. When there are $K$ nodes in the wireless network, the overhead for CSI estimation is given by ${\cal{O}}_{\rm CSI}^{\rm{mpnn}}\left(K^2\right)$.
\subsubsection{Signaling Overhead for Message Passing}To compute the graph convolution in a distributed manner, the nodes may exchange the messages ${\bf m}_{j,i}^{\left( n \right)}$ between each other. For an $N$-layer GNN, when there are $K$ nodes in the wireless network, the total signaling overhead of message passing is ${\cal{O}}_{\rm MP}^{\rm{mpnn}}\left(NK\left(K-1\right)\right)$. Alternatively, the $i$-th node can compute ${\bf m}_{j,i}^{\left( n \right)}$ locally rather than exchange them between the nodes. In this way, based on (\ref{message}), the $j$-th node only needs to broadcast its embedding and local states, i.e., ${\bf{e}}_j^{\left( n \right)}, {\mathbf{z}_j}$, $N$ rounds in one frame\footnote{The processing delay for updating embeddings between two communication rounds may not be negligible. This may lead to higher overhead and long latency. In Section V, we will illustrate how to address this issue by using an Air-MPRNN framework.}. Since there are $K$ nodes, the overhead for message passing is given by ${\cal{O}}_{\rm MP}^{\rm{mpnn}}\left(NK\right)$. Because the latter approach has lower signaling overhead, we use broadcast for message passing hereafter.

According to the above analysis, the total signaling overhead for the distributed execution of MPNN, denoted by $N_{\rm O}^{\rm mpnn}$, can be expressed as
\begin{equation}\label{MPNNoverhead}
N_{O}^{\rm mpnn} = K^2\delta_{\rm csi} + NK\delta_{\rm mp},
\end{equation}
where $\delta_{\rm csi}$ is the number of symbols used for channel estimation, and $\delta_{\rm mp}$ is the number of symbols used to encode the graph embedding and local features $\left\{ {{\bf{e}}_j^{\left( n \right)}, {\mathbf{z}_j}}\right\}$.

\section{Air-MPNN Design for the Distributed Power Allocation}
In the previous section, we have shown in (\ref{MPNNoverhead}) that the signaling overhead for CSI estimation and message passing grows fast as the network size increases, which may degrade the system performance significantly for large-scale networks. To reduce the signaling overhead, we develop the Air-MPNN framework for the power allocation problems in large-scale dynamic wireless networks.
\subsection{Air-MPNN Framework}
It is worth noting that the received power of the pilot signals from the interference links depends on the CSI of the interference links. Some interesting questions thus arise: \textit{Can we design an MPNN-based framework to utilize the power of the pilot signals from the interference links to aggregate the features of neighbor nodes and edges without knowing CSI of each interference link? Is it possible for all the nodes to broadcast the pilot simultaneously to accomplish the message passing and aggregation phases of the graph convolution operation?} In response to these questions, we propose the novel Air-MPNN framework such that each node only needs to select its transmit power ${\tilde p}_i^{\left( n \right)}$ according to its local features ${\bf{z}}_i$, and local embedding ${\bf e}_i^{\left( n \right)}$ during the message passing and aggregation phases. After that, all the nodes broadcast their pilots simultaneously. Each node can accomplish the graph convolution computation by evaluating the total power of the received pilots from the interference links. The graph convolution kernel of the proposed Air-MPNN framework is formally defined as
\begin{equation}\label{AirMPNNmessage}
\text{Pilot transmit power:}\quad {\tilde p}_i^{\left( n \right)}= \Phi \left( {{\bf e} _i^{\left( n \right)},{{\bf{z}}_i};\theta } \right),
\end{equation}
\begin{equation}\label{AirMPNNmessageagg}
\begin{split}
&\text{Air message passing and aggregation:}\\
&\tilde a_i^{\left( n \right)} = sum\left( {{\tilde p}_j^{\left( n \right)}{{\left| {{h_{j,i}}} \right|}^2},j \in {\cal N}\left( i \right)} \right),
\end{split}
\end{equation}
\begin{equation}\label{AirMPNNupdate}
\text{Update:}\quad {{\bf e}_i}^{\left( {n + 1} \right)} = U\left( {{{\bf e}_i}^{\left( {n} \right)},\tilde a_i^{\left( n \right)},{{\bf{z}}_i};\phi } \right),
\end{equation}
\begin{equation}
\text{Output:}\quad {p_i} = {\Omega \left( {{\bf {e}}_i^{\left( {N } \right)};\varphi } \right)}.
\end{equation}
\begin{figure*}[!t]
	\centerline{\includegraphics[width=0.8\textwidth]{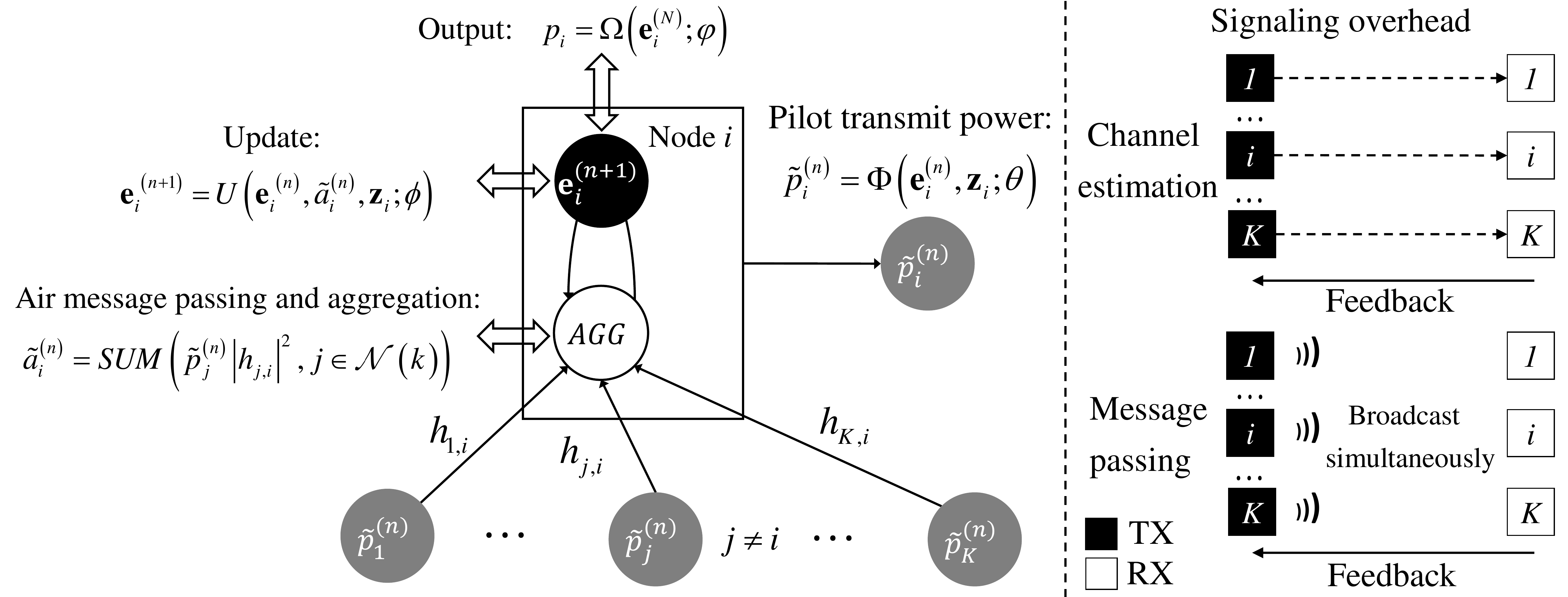}}
	\caption{The illustration of the proposed Air-MPNN framework, where all the nodes in each convolution layer require one round message passing.}
	\label{fig_airMPNN}
\vspace*{4pt}
\end{figure*}

For a better understanding, we depict an illustration of the graph convolution kernel for the proposed Air-MPNN framework in Fig. \ref{fig_airMPNN} on top of the next page. Compared with the graph convolution kernel of the existing MPNN framework given in (\ref{message})-(\ref{update}), the key difference of our proposed Air-MPNN framework is to replace the message passing and aggregation phases defined in (\ref{message}) and (\ref{agg}) with the air message passing and aggregation phase given in (\ref{AirMPNNmessageagg}), where the embedding and local features are represented by the transmit power of pilot in (\ref{AirMPNNmessage}). Specifically, in the existing MPNN framework, each node generates the message by considering its local state, local embedding, and the CSI of the interference link, i.e., ${\bf m}_{j,i}^{\left( n \right)} = \Phi \left( {{\bf {e}} _j^{\left( n \right)},{\mathbf{z}_j},\left| {{h_{j,i}}} \right|^2;\theta } \right)$. In the proposed Air-MPNN framework, each node generates the message only based on its local states and its embedding, i.e., ${\tilde p}_i^{\left( n \right)}= \Phi \left( {{\bf e} _i^{\left( n \right)},{{\bf{z}}_i};\theta } \right)$. From Fig. \ref{fig_MPNN} and Fig. \ref{fig_airMPNN}, the existing MPNN framework requires each node to broadcast message one by one for $N$ times to compute the graph convolution, while the proposed Air-MPNN framework only requires all the nodes to broadcast the pilot simultaneously for $N$ times. The aggregated messages in the Air-MPNN can be obtained by the following proposition.
\begin{proposition}
When all the nodes simultaneously broadcast orthogonal pilot sequences ${{{\bf{\tilde s}}}_i}$ with transmit power ${\tilde p}_i^{\left( n \right)}$, we can use the following equation to compute the aggregated message $\tilde a_i^{\left( n \right)}$ in (\ref{AirMPNNmessageagg}) and it is permutation invariant,
\begin{equation}\label{term}
\tilde a_i^{\left( n \right)} \approx {\left\| {{{{\bf{\tilde y}}}_i}} \right\|^2} - {\left| {{\bf{\tilde y}}_i^H{{{\bf{\tilde s}}}_i}} \right|^2},
\end{equation}
where $\tilde {\bf y}_i$ is the received signal at the $i$-th node.
\end{proposition}
\begin{proof}
The received signal can be expressed as
\begin{equation}\label{signalMPNN}
\tilde {\bf y}_i = \sqrt {{\tilde p}_i^{\left( n \right)}} {h_{i,i}}\tilde {\bf s}_i + \sum\limits_{j = 1,j \ne i}^K {\sqrt {{\tilde p}_j^{\left( n \right)}} {h_{j,i}}\tilde {\bf s}_j}  + {{\bf n}_i}.
\end{equation}
When the pilot sequences are orthogonal, we have ${\bf{\tilde s}}_i^H{{{\bf{\tilde s}}}_j} = 0,\forall i\ne j$, and ${\left\| {{{{\bf{\tilde y}}}_i}} \right\|^2}$ is given by
\begin{equation}\label{term1}
\begin{split}
{\left\| {{{{\bf{\tilde y}}}_i}} \right\|^2} & = {\bf{\tilde y}}_i^H{{{\bf{\tilde y}}}_i}\\
&\mathop  \approx \limits^{\left( a \right)} \left( {\sqrt {\tilde p_i^{\left( n \right)}} {{h}_{i,i}^*}{\bf{\tilde s}}_i^H + \sum\limits_{j = 1,j \ne i}^K {\sqrt {\tilde p_j^{\left( n \right)}} {{ h}_{j,i}^*}{\bf{\tilde s}}_j^H} } \right)\\
&\quad \times\left( {\sqrt {\tilde p_i^{\left( n \right)}} {h_{i,i}}{{{\bf{\tilde s}}}_i} + \sum\limits_{k = 1,k \ne i}^K {\sqrt {\tilde p_k^{\left( n \right)}} {h_{k,i}}{{{\bf{\tilde s}}}_k}} } \right)\\
&\mathop  = \limits^{\left( b \right)} \tilde p_i^{\left( n \right)}{\left| {{h_{i,i}}} \right|^2}{\bf{\tilde s}}_i^H{{{\bf{\tilde s}}}_i} + \sum\limits_{j = 1,j \ne i}^K {\tilde p_j^{\left( n \right)}{{\left| {{h_{j,i}}} \right|}^2}{\bf{\tilde s}}_j^H{{{\bf{\tilde s}}}_j}} \\
&\mathop  = \limits^{\left( c \right)} \tilde p_i^{\left( n \right)}{\left| {{h_{i,i}}} \right|^2} +  sum\left( {{\tilde p}_j^{\left( n \right)}{{\left| {{h_{j,i}}} \right|}^2},j \in {\cal N}\left( i \right)} \right).\\
\end{split}
\end{equation}
Approximation in (a) is accurate when noise power is negligible compared with signal power and interference power. It is the case in practical dense D2D networks, where interference is the bottleneck for achieving high data rate. For equality $\left(b\right)$, since ${\bf{\tilde s}}_i^H{{{\bf{\tilde s}}}_j} = 0$, $\forall i\ne j$, we only have ${\bf{\tilde s}}_i^H{{{\bf{\tilde s}}}_i}$ and ${\bf{\tilde s}}_j^H{{{\bf{\tilde s}}}_j}$. Expression $\left(c\right)$ is obtained under the assumption that the pilot sequences have unit power, i.e., ${\bf{\tilde s}}_i^H{{{\bf{\tilde s}}}_i}=1$.

In addition, ${\left| {{\bf{\tilde y}}_i^H{{{\bf{\tilde s}}}_i}} \right|^2}$ can be expressed as
\begin{equation}\label{term2}
\begin{split}
{\left| {{\bf{\tilde y}}_i^H{{{\bf{\tilde s}}}_i}} \right|^2} &\approx {\left| {\sqrt {\tilde p_i^{\left( n \right)}} {{ h}_{i,i}^*}{\bf{\tilde s}}_i^H{{{\bf{\tilde s}}}_i} + \sum\limits_{j = 1,j \ne i}^K {\sqrt {\tilde p_j^{\left( n \right)}} {{ h}_{j,i}^*}{\bf{\tilde s}}_j^H{{{\bf{\tilde s}}}_i}} } \right|^2}\\
& = {\left| {\sqrt {\tilde p_i^{\left( n \right)}} {{h}_{i,i}^*}} \right|^2} = \tilde p_i^{\left( n \right)}{\left| {{h_{i,i}}} \right|^2}.\\
\end{split}
\end{equation}
Substituting (\ref{term1}) and (\ref{term2}) into (\ref{term}), we can obtain that $\tilde a_i^{\left( n \right)} = sum\left( {{\tilde p}_j^{\left( n \right)}{{\left| {{h_{j,i}}} \right|}^2},j \in {\cal N}\left( i \right)} \right)$ in the interference limited network, where the noise power is much lower than the interference power. Since the operation $sum\left( {{\tilde p}_j^{\left( n \right)}{{\left| {{h_{j,i}}} \right|}^2},j \in {\cal N}\left( i \right)} \right)$ is a special case of $AGG\left( {{\bf m}_{j,i}^{\left( n \right)},j \in {\cal N}\left( i \right)} \right)$ in (\ref{agg}), it is permutation invariant.
\end{proof}
\begin{remark}
It is worth emphasizing that the $mean\left(\right)$ and $max\left(\right)$ aggregation functions can also be implemented in the proposed Air-MPNN framework. For the $mean\left(\right)$ function, each node needs to know the number of nodes in the network, which can be broadcast by a network manager. Then, the mean value is the ratio between the total interference power and the number of nodes. For the $max\left(\right)$ function, we assume that each node knows the pilot sequences codeword of all nodes prior to broadcasting, and all the pilot sequences are transmitted over the same time and frequency resource blocks. Since the pilot sequences are orthogonal, similar to (\ref{term2}), we have $ {\left| {{\bf{\tilde y}}_i^H{{{\bf{\tilde s}}}_k}} \right|^2}\approx\tilde p_k^{\left( n \right)}{\left| {{h_{k,i}}} \right|^2}$. Thus, the maximal interference power can be computed by $max\left( {{{\left| {{\bf{\tilde y}}_i^H{{{\bf{\tilde s}}}_k}} \right|}^2},k \in {\cal N}\left( i \right)} \right)$. We take summation as an example of the aggregation function in the implementation of Air-MPNN hereafter because of its simplicity.

In addition, it is worth noting that the proposed aggregation over-the-air technique in Proposition 1 is not only applicable to the GNN, but also applicable to the WMMSE algorithm, i.e., Algorithm 1 in \cite{DL1}, and we name it the Air-WMMSE. Similar to Proposition 1, during each iteration, line 7 and line 8 for Algorithm 1 in \cite{DL1} can be computed based on the received superimposed signal at each receiver, when all the transmitters broadcast orthogonal pilot sequences simultaneously. Since orthogonal pilot sequences are broadcast twice during each iteration, the signaling overhead can be high when the iteration number is large. To avoid high overhead in the Air-WMMSE, we can choose one iteration for practical implementation.
\end{remark}
\subsection{Signaling Overhead of Distributed Power Allocation via Air-MPNN}
The proposed Air-MPNN framework can be trained in a centralized manner with a similar procedure as the existing MPNN framework discussed below (\ref{problem2}) by replacing the MPNN kernel with the Air-MPNN kernel. We now analyze the overhead for implementing the Air-MPNN framework in a distributed manner. To compute (\ref{AirMPNNmessage})-(\ref{AirMPNNupdate}), each node needs to obtain the observation given by
\begin{equation}\label{observationAirMPNN}
{o_i^{\rm air-mpnn}} = \left\{ {{h_{i,i}},{w_i}, \left\{ {{{\tilde a}}_i^{\left( n \right)},\forall n} \right\}} \right\},
\end{equation}
where $\left\{{h_{i,i}},{w_i}\right\}$ is used to generate the pilot transmit power, and the aggregated message ${{\tilde a}}_i^{\left( n \right)}$ is used to update the local embedding for the $n$-th layer.
\subsubsection{Signaling Overhead for CSI Estimation}Since each node only requires the local CSI $h_{i,i}$, when there are $K$ nodes in the wireless network, the overhead for CSI estimation is given by ${\cal{O}}_{\rm CSI}^{\rm{air-mpnn}}\left(K\right)$.
\subsubsection{Signaling Overhead for Message Passing}According to Proposition 1 and Fig. \ref{fig_airMPNN}, all the nodes can obtain ${{\tilde a}}_i^{\left( n \right)}$ by broadcasting their pilot sequences in the same time-frequency resource block, and the Air-MPNN framework requires $N$ communication rounds for message passing. Since we assume the pilot sequences are orthogonal, the overhead is thus given by ${\cal{O}}_{\rm MP}^{\rm{air-mpnn}}\left(NK\right)$.

Because we only need to evaluate the total power of the interference links from the received pilots, the orthogonal pilot sequences for CSI estimation and message passing can be the same in Air-MPNN, its total signaling overhead is thus given by
\begin{equation}\label{AirMPNNoverhead}
N_{O}^{\rm air-mpnn} = K\delta_{\rm csi} + NK\delta_{\rm csi}=\left(N+1\right)K\delta_{\rm csi}.
\end{equation}
Compared with the existing MPNN framework, the signaling overhead is much lower since it grows linearly with network size $K$ and the number of layers $N$.

\section{Air-MPRNN Design for the Distributed Power Allocation}
Both the existing MPNN and the proposed Air-MPNN frameworks require to update the graph embedding $N$ times in order to obtain the transmit power for data transmission. It means that we need $N$ rounds of communications before data transmission. In addition, they need to perform channel estimation before message passing in order to obtain the local observation. The latency and signaling overhead are still unsatisfactory for some wireless networks. For example, the duration of each subframe in 5G New Radio is $1$~ms with a small number of symbols. It is not possible to transmit pilots multiple rounds before data transmission. To further reduce the latency and signaling overhead, we combine the Air-MPNN framework with RNN and develop an Air-MPRNN framework.
\subsection{Air-MPRNN Framework}
Considering that the frame duration of many wireless networks are smaller than the channel coherence time, e.g., 5G NR and Wi-Fi 802.11, the network state information across different frames is correlated, e.g., the CSI. Then, an interesting question arises: \textit{Can we utilize the graph embedding and CSI in the previous frame to update the graph embedding in the current frame?} If this is possible, we only need to update graph embedding once for each transmission frame. Meanwhile, we can combine the CSI estimation and message passing phases in the Air-MPNN framework, such that both the CSI and the aggregated message during the current frame can be evaluated based on the same received pilot signal. To achieve this, we propose the novel Air-MPRNN framework by exploring the temporal correlation of the wireless network. Inspired by the conventional RNN framework, we use the graph embedding in the previous frame as the hidden state in the RNN and pass it to the next frame for further graph convolution. Compared with MPNN and Air-MPNN with $N$ layers, the Air-MPRNN only has one layer and the graph embedding is updated once every frame. In the following, we remove the subscript $n$ for the GNN layer in designing the Air-MPRNN. Based on the graph convolution kernel of the Air-MPNN defined in (\ref{AirMPNNmessage})-(\ref{AirMPNNupdate}), we now formally define the recurrent graph convolution kernel of the proposed Air-MPRNN framework by
\begin{equation}\label{AirMPRNNmessage}
\begin{split}
\text{Pilot transmit power:}\quad {\tilde p}_i{\left(t\right)}= \Phi \left( {{\bf e} _i{\left( t-1 \right)},{{\bf{z}}_i\left( t-1 \right)};\theta } \right),
\end{split}
\end{equation}
\begin{equation}\label{AirMPRNNmessageagg}
\begin{split}
&\text{Air message passing and aggregation:}\\
&\tilde a_i{\left( t \right)} = sum\left( {{\tilde p}_j{\left(t\right)}{{\left| {{h_{j,i}}\left(t\right)} \right|}^2},j \in {\cal N}\left( i \right)} \right),
\end{split}
\end{equation}
\begin{equation}\label{AirMPRNNupdate}
\text{Hidden state update:}\quad {\bf e} _i{\left( t \right)} = U\left( {{\bf e} _i{\left( t-1 \right)},\tilde a_i{\left( t \right)},{{\bf{z}}_i\left( t \right)};\phi } \right),
\end{equation}
\begin{equation}\label{AirMPRNNoutput}
\text{Output:}\quad {p_i}\left( {t } \right) =  {\Omega \left( {{\bf e} _i{\left( t \right)};\varphi } \right)} .
\end{equation}
\begin{figure*}[!t]
	\centerline{\includegraphics[width=0.8\textwidth]{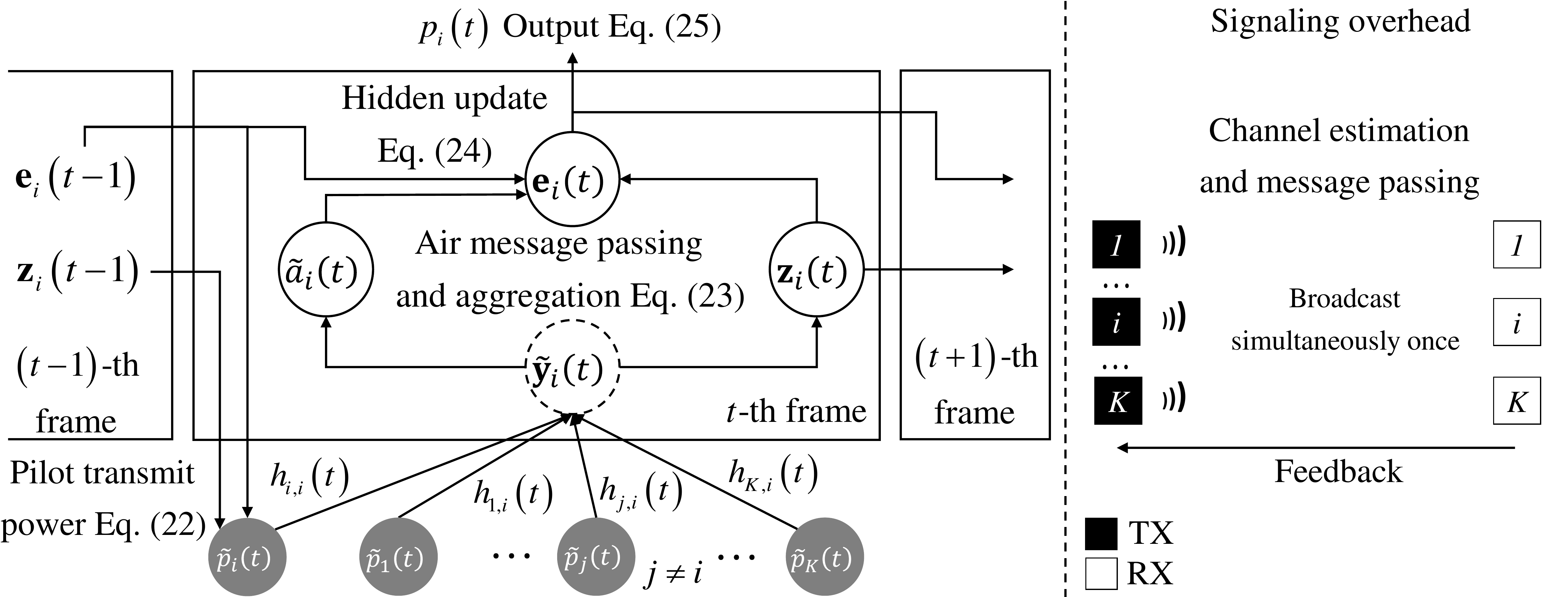}}
	\caption{The illustration of the proposed Air-MPRNN framework.}
	\label{fig_airMPRNN}
\vspace*{4pt}
\end{figure*}

For a better understanding, we depict an illustration of the recurrent graph convolution kernel of the proposed Air-MPRNN framework in Fig. \ref{fig_airMPRNN} on top of the next page. In contrast to the proposed Air-MPNN framework shown Fig. \ref{fig_airMPNN}, where the nodes need to broadcast pilot multiple times to compute the graph convolution, the proposed Air-MPRNN framework only requires all the nodes to broadcast pilot simultaneously once. To be more specific, at the beginning of the $t$-th frame, the $i$-th node takes the local state information ${\bf{z}}_i \left(t-1\right) = \left[ {\left|{h_{i,i}\left(t-1\right)}\right|^2,{w_i}} \right]$ and the graph embedding, i.e., hidden state, ${\bf e} _i{\left( t-1 \right)}$ updated in the previous frame to generate the pilot transmit power ${\tilde p}_i{\left(t\right)}$ according to (\ref{AirMPRNNmessage}). After that, all the nodes broadcast pilot sequences with a transmit power ${{\tilde p}_i\left(t\right)}$ and receive the signal ${{{\bf{\tilde y}}}_i}\left(t\right)$ given by
\begin{equation}\label{signalMPRNN}
\tilde {\bf y}_i\left(t\right) = \sqrt {{\tilde p}_i\left(t\right)} {h_{i,i}\left(t\right)}\tilde {\bf s}_i + \sum\limits_{j = 1,j \ne i}^K {\sqrt {{\tilde p}_j\left(t\right)} {h_{j,i}\left(t\right)}\tilde {\bf s}_j}  + {{\bf n}_i}\left(t\right).
\end{equation}
Each node can then utilize the received signal ${{{\bf{\tilde y}}}_i}\left(t\right)$ to obtain ${\bf{z}}_i \left(t\right)$ and $\tilde a_i{\left( t \right)}$ by the following proposition.
\begin{proposition}
In the Air-MPRNN framework, when the pilot sequences are orthogonal, i.e., ${\bf{\tilde s}}_i^H{{{\bf{\tilde s}}}_j} = 0,\forall i\ne j$, both the local state ${\bf{z}}_i \left(t\right)$, i.e., $\left|{h_{i,i}\left(t\right)}\right|^2$, and the aggregated message $\tilde a_i{\left( t \right)}$ in the $t$-th frame can be obtained from $\tilde {\bf y}_i\left(t\right)$ by the following equations, and the operation is permutation invariant,
\begin{equation}\label{term11}
\left|h_{i,i}\left(t\right)\right|^2 \approx {\left| {{\bf{\tilde y}}_i^H\left(t\right){{{\bf{\tilde s}}}_i}} \right|^2}/{{\tilde p}_i{\left( t \right)}},
\end{equation}
\begin{equation}\label{term22}
\tilde a_i{\left( t \right)} \approx {\left\| \tilde {\bf y}_i\left(t\right) \right\|^2} - {\left\| \tilde {\bf y}_i^H\left(t\right)\tilde {\bf s}_i \right\|^2}.
\end{equation}
\end{proposition}
\begin{proof}
It is worth noting that different from the Air-MPNN, the proposed Air-MPRNN uses the local state in the previous frame to generate the pilot transmit power ${{\tilde p}_i\left(t\right)}$ in the current frame. Therefore, each node in Air-MPRNN does not need the CSI $h_{i,i}\left(t\right)$ in the current frame to determine ${{\tilde p}_i\left(t\right)}$ and can receive $\tilde {\bf y}_i\left(t\right)$ directly. (\ref{term11}) can thus be obtained based on (\ref{term2}). Like Proposition 1, we can derive that ${\left\| \tilde {\bf y}_i\left(t\right) \right\|^2} - {\left\| \tilde {\bf y}_i^H\left(t\right)\tilde {\bf s}_i \right\|^2} \approx SUM\left( {{\tilde p}_j{\left(t\right)}{{\left| {{h_{j,i}}\left(t\right)} \right|}^2},j \in {\cal N}\left( i \right)} \right)$ and it is permutation invariant. The proof follows.
\end{proof}

With Proposition 2, each node can update its local embedding according to (\ref{AirMPRNNupdate}), saves the updated embedding ${\bf e} _i{\left( t \right)}$ and the obtained local state ${\bf{z}}_i \left(t\right)$ for the next frame. Finally, the RNN unit outputs the transmit power according to (\ref{AirMPRNNoutput}). Compared with the graph convolution kernel of the Air-MPNN framework given in (\ref{AirMPNNmessage})-(\ref{AirMPNNupdate}), the recurrent graph convolution kernel of the Air-MPRNN framework utilizes the graph embedding and CSI in the previous frame to update embeddings recurrently. In this way, all the nodes only require to transmit the pilot signal once during each frame.
\subsection{Signaling Overhead of Distributed Power Allocation via Air-MPRNN}
The above power allocation policy is trained in a centralized manner. We now define the observation of each node if the policy is executed in a distributed manner. As shown in Fig.~\ref{fig_airMPRNN}, to compute (\ref{AirMPRNNmessage})-(\ref{AirMPRNNupdate}), the $i$-th node only needs to obtain the following observation
\begin{equation}\label{observationAirMPNN}
{o_i}\left( t \right) = \left\{ {{{{\bf{\tilde y}}}_i}\left( t \right),{w_i}\left( t \right)} \right\}.
\end{equation}
Since we need to assign orthogonal pilots to the $K$ D2D pairs to both evaluate CSI and the aggregated messages from ${{{\bf{\tilde y}}}_i}\left( t \right)$, the overall signaling overhead for the proposed Air-MPRNN model can be expressed as
\begin{equation}\label{AirMPRNNoverhead}
N_{O}^{\rm air-mprnn} = K\delta_{\rm csi}.
\end{equation}
\begin{remark}
The Air-MPRNN can be potentially integrated into the existing wireless networks without changing the frame structures in the standards, such as 5G NR and Wi-Fi 802.11 systems. By implementing the message, update, and output MLPs at each transmitter, the transmitters can adjust the transmit power of the pilot sequence according to (\ref{AirMPRNNmessage}) and broadcast the pilot simultaneously. Each receiver can update $\tilde a_i{\left( t \right)}$ and ${{\bf{z}}_i\left( t \right)}$ by Proposition 2 based on the received superimposed signal. By feeding back $\tilde a_i{\left( t \right)}$ and ${{\bf{z}}_i\left( t \right)}$ to the transmitter via CSI feedback, each transmitter can update its hidden state according to (\ref{AirMPRNNupdate}), and determine the frame transmit power by (\ref{AirMPRNNoutput}). Different from AirComp systems, the aggregation over-the-air mechanism in Proposition 2 only requires to evaluate the summation of the interference power and there is no need to obtain the signal from each link. Furthermore, the synchronization of all the transmitters in the considered D2D network can be achieved by synchronizing them with a base station as specified in the existing 5G NR \cite{sync}. Therefore, the symbol-level synchronization is not hard to achieve in our system.
\end{remark}

\section{Simulation Results}
In this section, we evaluate the performance of the proposed Air-MPNN and Air-MPRNN frameworks, and compare them with the existing MPNN, WMMSE, and EPA algorithms.
\subsection{Simulation Setup}
\subsubsection{Layout Generation}To generate the training data, we consider a $500$ m $\times$ $500$ m layout, where $20$ D2D pairs are randomly located in this area. We define the link density as $\beta  = {{K} \over {l^2}}$, where $K$ is the number of D2D pairs and $l$ is the field length (m). In addition, in order to capture different link densities, we define the link density factor as the ratio of link densities between the training stage and the testing, i.e., $\gamma  = {{{\beta _{{\rm{train}}}}} \over {{\beta _{{\rm{test}}}}}}$. In Subsections C-E, we set the density factor to one and illustrate the impact of other parameters on the performance of different policies. In Subsection F, we will investigate the impact of the link density factor.
\subsubsection{CSI Generation}
\begin{table*}[t!]
\scriptsize
\centering
\caption{System parameters for simulation setup}\label{wireless network_parameters}
 \begin{tabular}{||c | c || c |c||}
 \hline
 Parameters & Values  &Parameters &Values\\ [0.5ex]
 \hline\hline
 Bandwidth & 5 M  & Transmit antenna power gain & 2.5 dBi\\
 \hline
 Carrier frequency & 2.4 GHz & Maximum transmit power & 40 dBm \\
 \hline
 Transmit antenna height & 1.5 m  & Receive antenna height & 1.5 m \\
 \hline
 D2D link density & $8 \times 10^{-5}$ links/m$^2$ &  Channel correlation coefficient & [0 1) \\
 \hline
 Number of training layouts & 2000 & Number of testing layouts & 500 \\
 \hline
 Number of frames & 10 & Number of symbols in frame & 3000\\
 \hline
 Number of nodes & 20 & Noise power & -169 dBm/Hz\\
 \hline
CSI estimation overhead $\delta_{\rm csi}$ & 1 symbol & Message passing overhead $\delta_{\rm mp}$  & 5 symbols \\[1ex]
\hline
 \end{tabular}
\end{table*}
Based on the generated layout, we then generate the channel coefficients and channel power gains of the wireless links. To capture the path-loss, we consider the following model in dB for the ultra high frequency (UHF) propagation given by \cite{path-loss}
\begin{equation}\label{path_loss}
{L_{pl}} = {L_{bp}} + 6 + \left\{ {
\begin{matrix}
\begin{split}
   &{20{{\log }_{10}}\left( {{d \over {{R_{bp}}}}} \right),{\rm{ for }}\quad d \le {R_{bp}}}  \\
   &{40{{\log }_{10}}\left( {{d \over {{R_{bp}}}}} \right),{\rm{ for }}\quad d > {R_{bp}}}  \\
\end{split}
\end{matrix} } \right.,
\end{equation}
where $d$ (m) is the distance between the transmitter and the receiver, ${R_{bp}} = 4{h_1}{h_2}/\lambda $ is the breakpoint distance, $h_1$, $h_2$ are the heights of the transmission and receiving antennas, and $\lambda$ is the wavelength. ${L_{bp}} = \left| {20{{\log }_{10}}\left( {{{{\lambda ^2}} \over {8\pi {h_1}{h_2}}}} \right)} \right|$ in (\ref{path_loss}) is the basic path-loss at the breakpoint \cite{path-loss}.

For each layout, we let $g_{ls}$ denote the large-scale fading path gain obtained from the path-loss model (\ref{path_loss}), and let $h_{ss}\left(t\right)$ represent the small-scale channel fading coefficient. To capture the channel correlation, the channel gain of each link evolves according to \cite{channelmodel}
\begin{equation}\label{smallscale1}
h_{ss}\left(t+1\right) = \rho h_{ss}\left(t\right)+\beta\left(t\right),
\end{equation}
\begin{equation}\label{smallscale2}
\left|h\left(t\right)\right|^2 = g_{ls}\left|h_{ss}\left(t\right)\right|^2,
\end{equation}
where $\beta\left(t\right) \sim {\cal{CN}}\left(0,1-\rho^2\right)$, and $\rho$ is the channel correlation coefficient. Note that $\rho = 0$ corresponds to the special case that the small-scale channel fading coefficient changes independently across different frames. We uniformly choose $\rho \in \left[0\right.\left.,1\right)$ for each layout in generating all the training and testing data sets. For the testing data set, we choose the same system parameters as the training data set except for the number of layouts. All the system parameters are summarized in Table \ref{wireless network_parameters} on the top of this page, unless otherwise specified.
\subsubsection{Normalization}
We next discuss the data normalization for the considered frameworks. In the existing MPNN, both the local CSI of each D2D link, and the CSI of each interference link can be normalized by $\tilde{x} = \left( {x - \overline x } \right)/\widehat x$, where $x$ is the training data, $\overline x$ is the average value of the data, and $\widehat x$ is the standard deviation of the data. However, in the proposed Air-MPNN and Air-MPRNN frameworks, we need to normalize the aggregated interference power, i.e., $\tilde \alpha _i^{\left( n \right)}$ in (\ref{AirMPNNmessageagg}) and $\tilde \alpha _i^{\left( t \right)} $ in (\ref{AirMPRNNmessageagg}). The problem is that we do not have samples for such data and its value varies when the weights of the GNN change. In our simulation, we adopt constant values $\overline  m$ and $\widehat m$ to normalize $\tilde \alpha _i^{\left( n \right)}$ and $\tilde \alpha _i^{\left( t \right)}$ by $\tilde{x} = \left( {x - \overline m} \right)/\widehat m$, where $x$ is the data, $\bar m$, $\widehat m$ are the mean value and the standard deviation of the channel power gain of the interference links obtained from the training data set, respectively.
\subsection{GNN Architecture \& Training}
During each training iteration, a batch of $50$ samples are randomly selected from the training data set, and fed to the GNN frameworks to update their weights in an unsupervised manner according to the loss function defined as the negative of the objective function. We set the weights for all the nodes equal to 1, i.e., $w_k = 1$, $k=1,2,\cdots,K$ and consider the sum-rate performance hereafter. We tried different hyper-parameters and select the best one by trial and error. The dimension of graph embedding of each node in all the GNN-based frameworks is set to $8$. Recall that all the GNN-based frameworks contain a message MLP $\Phi \left( {\cdot;\theta } \right)$, an update MLP $U \left( {\cdot;\phi } \right)$ and an output MLP ${\Omega \left( {\cdot;\varphi } \right)}$, we list their structures in Table \ref{learning_parameters} on the top of next page. It is worth noting that the output dimension of the message MLPs in Air-MPNN and Air-MPRNN is set to 1 since it is represented by the pilot transmit power. Differently, the output dimension of the message MLP in MPNN is set to 32 for a richer representation.
\begin{table*}[t!]
\scriptsize
\centering
\caption{Hyper-Parameters for MPNN, Air-MPNN and Air-MPRNN.}\label{learning_parameters}
 \begin{tabular}{||c | c|| c | c ||}
 \hline
 Hyper-Parameters & Values & Hyper-Parameters & Values \\ [0.5ex]
 \hline\hline
 Initial learning rate & 0.002 & Learning rate decay factor& 0.9\\
 \hline
 Batch size & 50 & Dimension of graph embedding & 8  \\
 \hline
 Optimizer & Adam \cite{ADAM} & Structure of Output MLP & $\left\{8,16,1\right\}$  \\
 \hline
 \multirow{3}*{Graph convolution layers} & MPNN: 3& \multirow{3}*{Model parameters} & MPNN: 2377\\
 \cline{2-2}\cline{4-4}
 ~ &Air-MPNN: 3 & ~& Air-MPNN: 1882    \\
 \cline{2-2}\cline{4-4}
 ~ &Air-MPRNN: 1& ~& Air-MPRNN: 2258 \\
 \hline
 \multirow{3}*{Structure of Message MLP} & MPNN: $\left\{10,32,32\right\}$ & \multirow{3}*{Structure of Update MLP}& MPNN :$\left\{41,16,8\right\}$\\
 \cline{2-2}\cline{4-4}
 ~ &Air-MPNN: $\left\{9,32,32,1\right\}$ & ~& Air-MPNN: $\left\{10,16,8\right\}$    \\
 \cline{2-2}\cline{4-4}
 ~ &Air-MPRNN: $\left\{9,32, 32, 1\right\}$ & ~& Air-MPRNN: $\left\{10,32,8\right\}$ \\ [1ex]
 \hline
 \end{tabular}
\end{table*}
\begin{figure}[!h]
	\centerline{\includegraphics[width=0.45\textwidth]{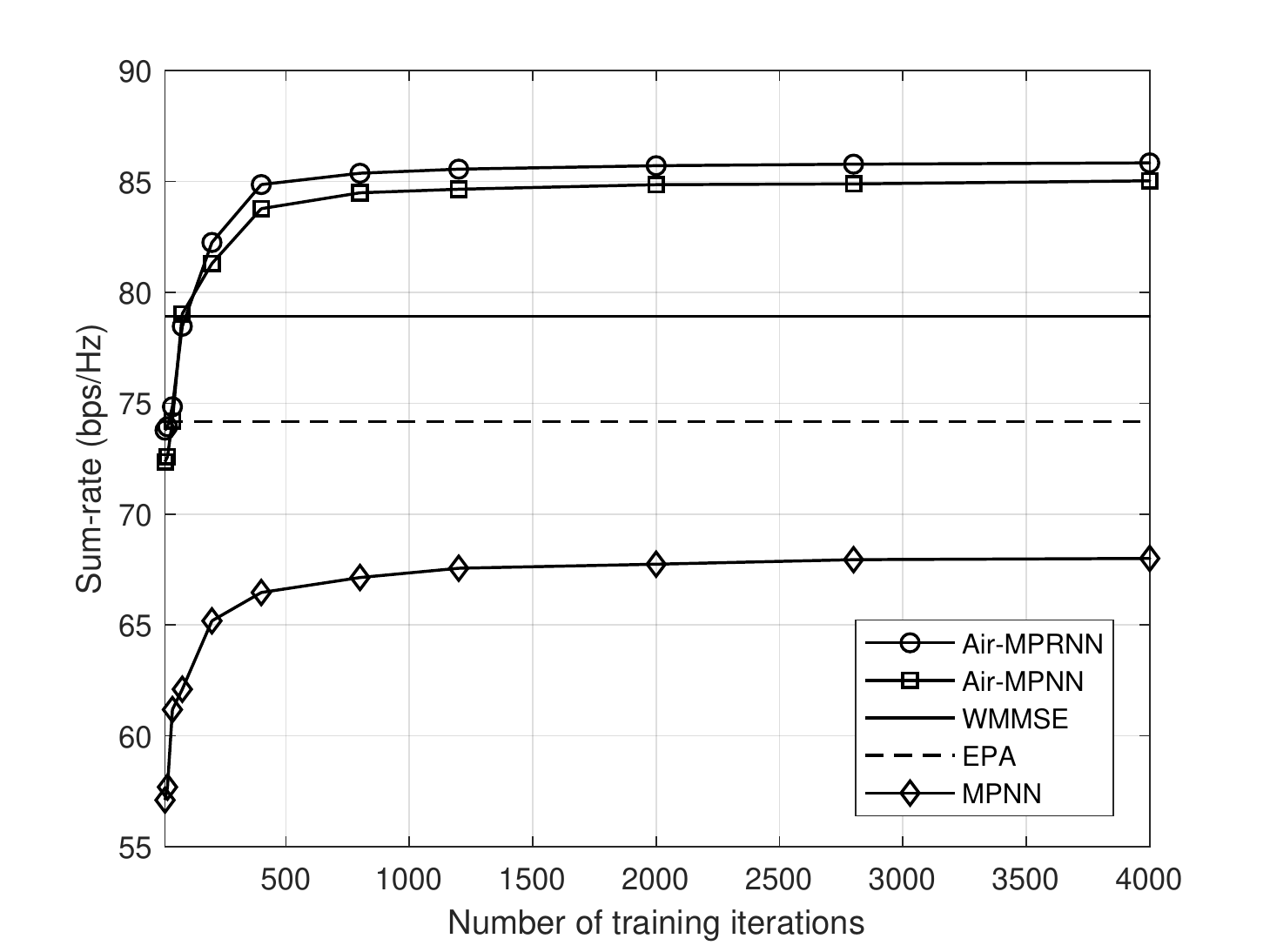}}
	\caption{The sum-rate of the MPNN, Air-MPNN and Air-MPRNN frameworks during training.}
	\label{fig_training_performance}
\end{figure}

In Fig. \ref{fig_training_performance}, we first depict the performance of the considered frameworks during training, by evaluating their sum-rate on the same validation data set against different number of iterations. For comparison, we also show the sum-rate of the conventional EPA and WMMSE algorithms. In EPA, we assume that all the nodes select the maximum transmit power for data transmission. In WMMSE, the transmit power is obtained after 100 iterations based on the global CSI. In EPA, there is no signaling overhead. With WMMSE, we need to estimate CSI of all the $K^2$ links and thus the signaling overhead is $K^2\delta_{\rm csi}$. From Fig. \ref{fig_training_performance}, we can observe that the proposed Air-MPNN and Air-MPRNN frameworks can outperform the existing EPA and WMMSE algorithms after several hundreds of training iterations. Furthermore, we can see that the existing MPNN performs poorly due to the signaling overhead for CSI estimation and message passing. As shown in Table~\ref{wireless network_parameters}, the CSI estimation and message passing overheads are $\delta_{\rm csi} = 1$ and $\delta_{\rm mp} = 5$ symbols. A total number of $N_{O}^{\rm mpnn} = K^2\delta_{\rm csi} + NK\delta_{\rm mp} = 700$ symbols are required to obtain the observation for MPNN. There are 3000 symbols in each frame, and hence the overhead is not negligible. This observation implies that the existing MPNN framework is not suitable for distributed execution. Because the sum-rate of the considered frameworks can converge with 2000 iterations, we thus set the number of training iterations to 2000 for all the frameworks. We perform the training and testing of the GNN-based frameworks on a Laptop with NVIDIA GeForce RTX 3050 Ti Laptop GPU. Because the WMMSE algorithm has a sequential computation flow, it is running on the Intel(R) Core(TM) i7-11370H CPU (3.30Hz) on the same laptop. The training time, sum-rate, and the signaling overhead ratio $N_O \over N_S$ for all the algorithms are shown in Table~\ref{Time_Comparison}. From Table~\ref{Time_Comparison}, we can see that the training time for all the considered three GNN-based frameworks is less than 1 min. After the training phase, the Air-MPRNN framework achieves the highest sum-rate, followed by the Air-MPNN. They both outperforms the conventional EPA and WMMSE algorithms. The existing MPNN achieves the lowest sum-rate due to its high signaling overhead.
\begin{table}[h!]
\scriptsize
\centering
\caption{Comparison of training for the considered algorithms.}\label{Time_Comparison}
 \begin{tabular}{||c | c| c | c  ||}
 \hline
 Algorithm & Training time (s) & Sum-rate (bps/Hz) & Overhead ratio ${ N_O} \over {N_S}$\\ [0.5ex]
 \hline\hline
 EPA & - & 74.17 & 0\\
 \hline
 WMMSE &-& 78.88 & 13.3\%\\
 \hline
 MPNN & 14.84 & 67.26 & 23.3\%  \\
 \hline
 Air-MPNN & 16.41 & 84.80  & 2.7\%  \\
 \hline
 Air-MPRNN & 54.53 & 85.76  & 0.7\%  \\[1ex]
 \hline
 \end{tabular}
\end{table}
\subsection{Impact of signaling overhead}
\begin{figure}[!h]
	\centerline{\includegraphics[width=0.45\textwidth]{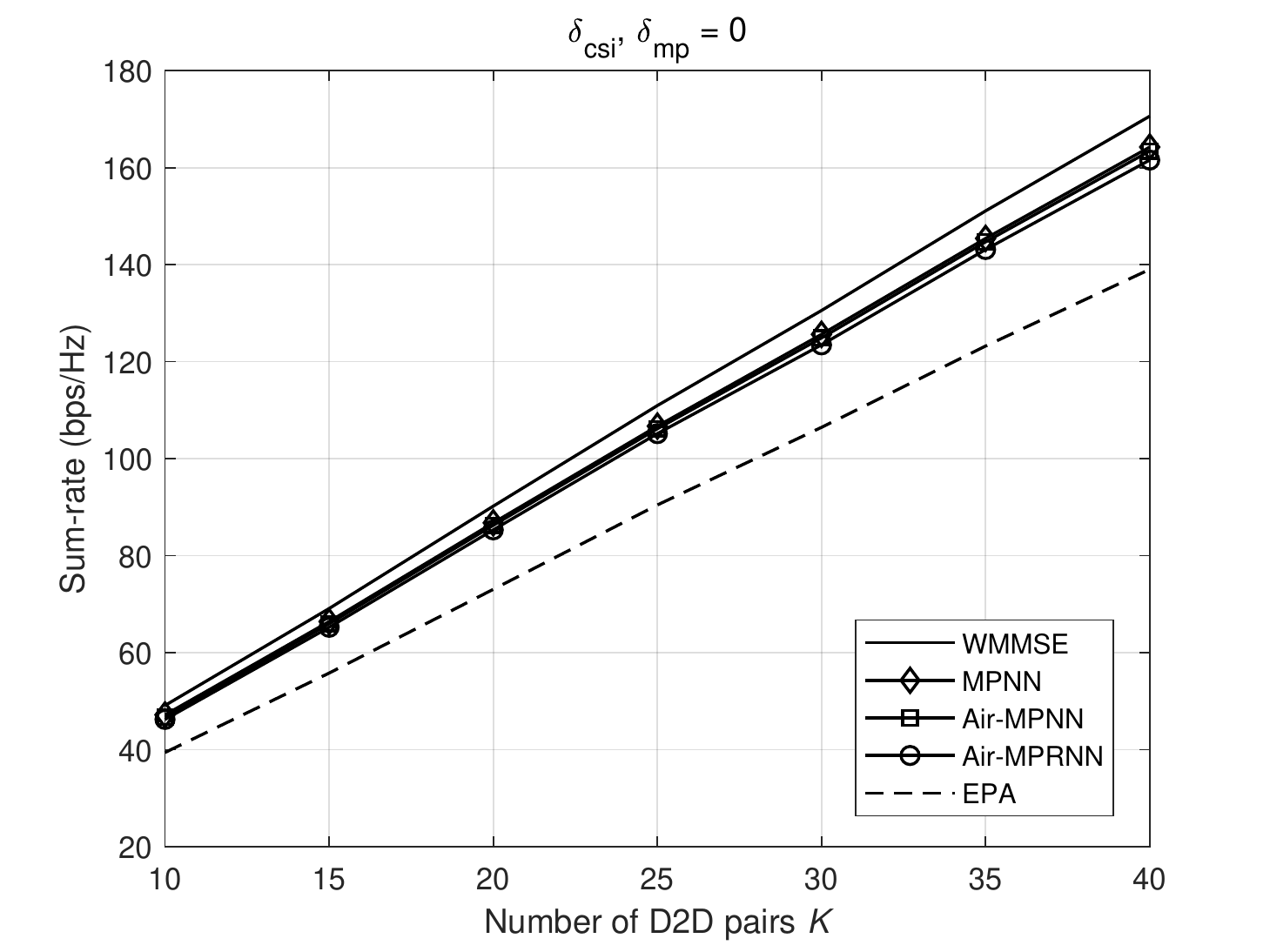}}
	\caption{The sum-rate versus the number of D2D pairs where the signaling overhead is assumed to be zero.}
	\label{fig_test_performance1}
\end{figure}
In Fig. \ref{fig_test_performance1}, we plot the sum-rate of the network versus the number of D2D pairs, where the signaling overhead is assumed to be zero, i.e., $\delta_{\rm csi}, \delta_{\rm mp} = 0$. It can be seen that under this assumption, the WMMSE algorithm achieves the highest sum-rate, while the EPA achieves the lowest sum-rate, respectively. All the three GNN-based frameworks can achieve near-optimal sum-rate compared with the WMMSE algorithm.
\begin{figure}[!h]
	\centerline{\includegraphics[width=0.45\textwidth]{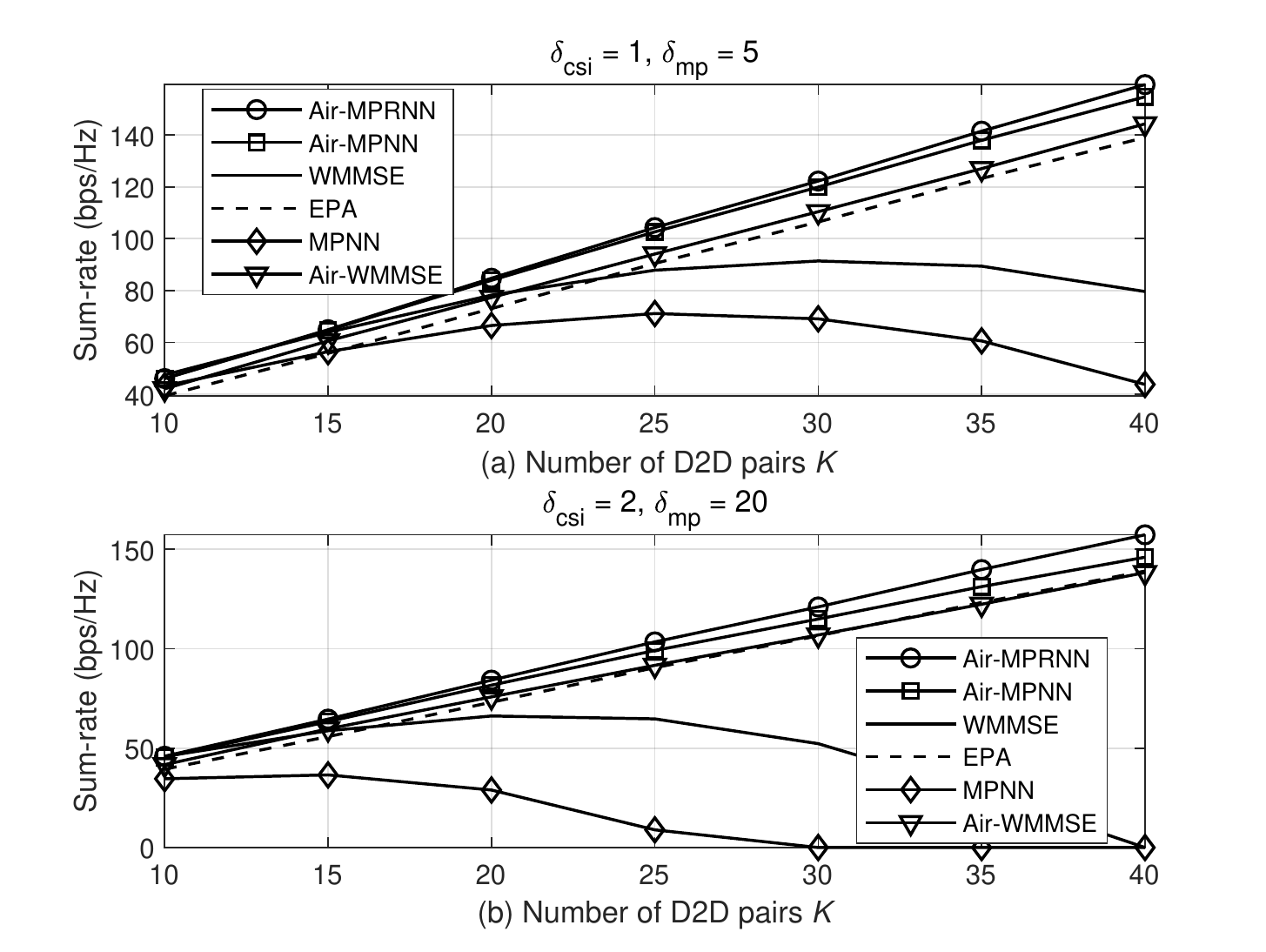}}
	\caption{The sum-rate versus the number of D2D pairs.}
	\label{fig_test_performance2}
\end{figure}

In Fig. \ref{fig_test_performance2}, we depict the relationship between sum-rate and the number of D2D pairs with different signaling overhead parameters. Recall that in Sec. IV-A, we have discussed the implementation of the proposed message passing over-the-air technique in WMMSE policy, namely the Air-WMMSE policy. The total signaling overhead of Air-WMMSE with one iteration is given by $3K\delta_{\rm csi}$, where $K\delta_{\rm csi}$ is the overhead for channel estimation of local D2D links, and $2K\delta_{\rm csi}$ is the overhead for broadcasting pilot twice during one iteration. The results in Fig. \ref{fig_test_performance2}(a) show that when $\delta_{\rm csi} = 1$, $\delta_{\rm mp} = 5$, the sum-rate of the WMMSE and MPNN firstly increases as the network size grows, then decreases when the number of D2D pairs is high. This is because that they both require the CSI of the D2D links and the interference links. The results in Fig. \ref{fig_test_performance2}(b) show that when $\delta_{\rm csi} = 2$, $\delta_{\rm mp} = 20$, the existing MPNN achieves zero sum-rate when there are 30 D2D pairs. When $K\ge30$, the symbols required for CSI estimation and message passing is larger than the total number of symbols in one frame. From the results in Fig. \ref{fig_test_performance2}, we can see that the proposed Air-MPNN and Air-MPRNN frameworks scale with the network size. Therefore, the proposed Air-MPNN and Air-MPRNN frameworks are suitable for large-scale wireless networks, where the existing MPNN and WMMSE may fail. Finally, the sum-rate of the Air-WMMSE is significantly higher than the existing WMMSE, which reveals the effectiveness of our aggregation over-the-air mechanism.
\begin{table*}[t!]
\scriptsize
\centering
\caption{The impact of channel correlation coefficient $\rho$ on the sum-rate of Air-MPRNN.}\label{channel_correlation}
 \begin{tabular}{||c | c| c | c |c |c|c |c||}
 \hline
 Channel correlation coefficient & 0 & 0.2 & 0.4 & 0.6 & 0.8 & 0.9 & 0.99\\ [0.5ex]
 \hline
 Normalized sum-rate of Air-MPRNN & 101.78\% & 101.84\% & 101.85\% & 101.95\% & 102.14\%& 102.34\% &102.80\%  \\[1ex]
 \hline
 \end{tabular}
\end{table*}
\subsection{Impact of frame duration}
\begin{figure}[!h]
	\centerline{\includegraphics[width=0.45\textwidth]{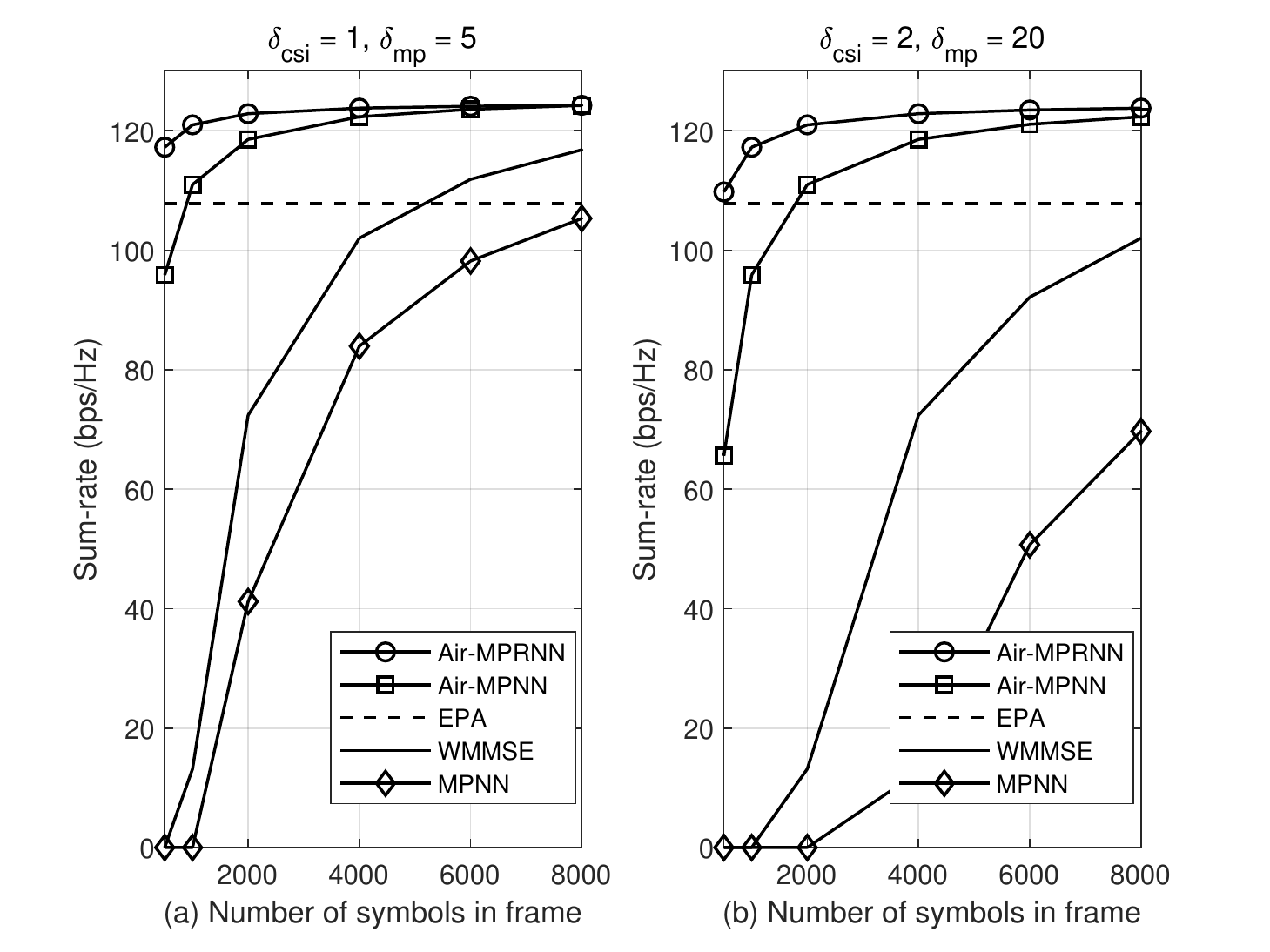}}
	\caption{The sum-rate versus the number of symbols in one frame, $N_S$, where the number of D2D pairs is $K =30$.}
	\label{fig_test_performance3}
\end{figure}
We next examine the impact of frame duration on the sum-rate, which is proportional to the number of symbols in the frame. We observe from Fig. \ref{fig_test_performance3} that as the frame duration increases, the sum-rate of EPA algorithm remains the same because it has no signaling overhead. The sum-rate of all the other algorithms increase as the frame duration increases since the overhead ratio $N_O \over N_S$ decreases. In addition, we can see that both the WMMSE and MPNN achieve poor sum-rate when $N_S$ is small, e.g., below 4000 symbols. Different from all the other GNN frameworks, the Air-MPRNN framework can achieve higher sum-rate than EPA when the frame duration is extremely short, e.g., a few hundreds symbols. This is because that the proposed recurrent graph convolution kernel only requires to broadcast message once during each frame, such unique feature becomes significant when the frame duration is very limited.
\subsection{Impact of channel correlation coefficient}
Since Air-MPRNN relies on the temporal correlation of network states, we examine the impact of correlation coefficient, $\rho$, on the sum-rate of Air-MPRNN. We set fixed values of $\rho$ for each layout in generating the testing data in this subsection rather than randomly choose $\rho$ between 0 and 1 in other subsections. The sum-rate of Air-MPRNN is normalized by the sum-rate of Air-MPNN. We choose the same parameters as that in Fig. \ref{fig_test_performance2}(a) where $K=30$. From Table \ref{channel_correlation} on top of this page, we can observe that as the channel correlation coefficient increases, the normalized sum-rate of Air-MPRNN slightly increases. It is surprising that the Air-MPRNN framework can achieve good performance even when $\rho=0$. This is because that even when the small-scale channel fading coefficient is not correlated and the Air-MPRNN utilizes the local CSI in the previous frame to update the graph embedding, it can still retain a good performance by aggregating the CSI of the interference links in the current frame.
\subsection{Impact of link density}
\begin{figure}[!h]
	\centerline{\includegraphics[width=0.45\textwidth]{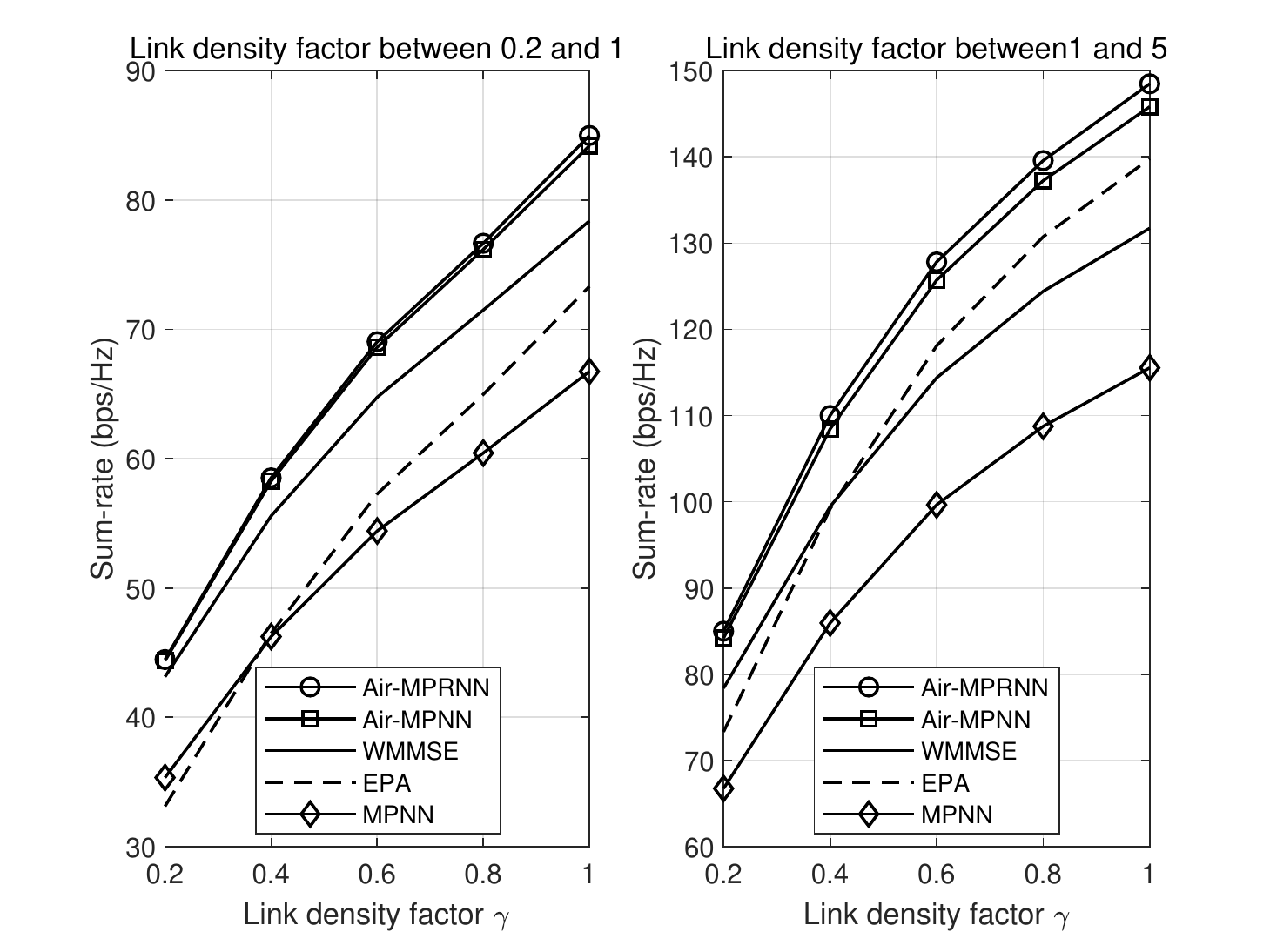}}
	\caption{The sum-rate versus the link density factor $\gamma$, where $K=20$, $\delta_{\rm csi} = 1$, $\delta_{\rm mp} = 5$.}
	\label{fig_test_density}
\end{figure}
We show the relationship between the sum-rate and the link density factor in Fig. \ref{fig_test_density}. Recall that a higher value of link density factor implies that the D2D pairs are less dense in the testing data compared with training data. We can first observe that the EPA achieves low sum-rate when the density is high, but it outperforms the WMMSE when the density is low. This is because that the EPA selects the maximum transmit power and it is the optimal policy when the interference between all the D2D pairs is negligible. Furthermore, the proposed Air-MPNN and Air-MPRNN frameworks can outperform the existing algorithms and scale well with the link density factor. These results validate the scalability of the proposed frameworks.


%
\section{Conclusions}
In this paper, we investigated the power allocation problem in dense wireless networks with multiple transceiver pairs. We developed GNN-based frameworks for distributed power allocation by taking the signaling overhead in terms of CSI estimation and message passing into account. Taking the sum-rate maximization problem as an example, we first analyzed the signaling overhead of the existing MPNN framework. Our results showed that both the CSI estimation and message passing overhead can grow quadratically as the network size increases. Inspired from AirComp, we then proposed a novel Air-MPNN framework by passing and aggregating the messages over-the-air. Specifically, each node can determine its pilot transmit power based on its embedding and local state, and broadcasts the pilot signal simultaneously. The messages, including the local embedding, local state and the CSI of the interference links, can be aggregated efficiently by evaluating the total received power of all the interference links. The overall signaling overhead of the proposed Air-MPNN grows linearly as the network size increases. Based on the Air-MPNN, we further proposed the Air-MPRNN by introducing the RNN into the framework. Different from Air-MPNN, the Air-MPRNN can exploit the temporal information of the wireless networks. It utilizes the graph embedding and CSI in the previous frame to update the graph embedding in the current frame. The Air-MPRNN can be potentially implemented in the existing wireless networks without changing the frame structures in standards by sending one pilot during each frame. It was shown that the proposed frameworks outperform the existing algorithms in terms of signaling overhead and sum-rate for various system parameters.
\subsection{Guidelines and Future Work}
The proposed GNN-based frameworks are particularly suitable for large-scale networks because the signaling overhead is limited, and they are scalable to the network size. In addition, the Air-MPNN framework can be applied to the first frame of a new scenario where historical information is not available, then we can apply the proposed Air-MPRNN framework for the following frames. It is also worth noting that the proposed GNN-based frameworks may not be robust when the wireless environment changes dramatically, while the conventional optimization-based methods can be robust for arbitrary wireless environment at the cost of complicated channel training and iterations. Finally, the existing centralized and offline training method for the proposed GNN-based frameworks may not obtain the optimal power allocation scheme in dynamic wireless networks. Developing a distributed in-network training method such that the neural parameters can be tuned online according to the changes of wireless environment remains an open problem and deserves further investigation.
\appendices
\ifCLASSOPTIONcaptionsoff
  \newpage
\fi

\bibliographystyle{IEEEtran}
\bibliography{References}
\end{document}